\documentclass[11pt]{article}
\usepackage{fullpage,times}

\usepackage[utf8]{inputenc}
\usepackage[T1]{fontenc}

\usepackage{amsmath}
\usepackage{amssymb}

\usepackage{amstext}
\usepackage{amsfonts}
\usepackage{amsthm}
\usepackage{mathtools}
\usepackage{scrextend}    
\usepackage{array}     
\usepackage{multicol}    
\usepackage{mathrsfs}    
\usepackage{color}
\usepackage{hyperref}
\usepackage{cleveref}
\usepackage{bm}
\usepackage{enumitem}
\usepackage{upgreek}
\usepackage{IEEEtrantools}
\usepackage{mleftright}
\usepackage[mathscr]{euscript}
\usepackage{thm-restate}

\newtheorem{theorem}{Theorem}[section]    
\newtheorem{definition}{Definition}[section] 
\newtheorem{corollary}[theorem]{Corollary}    
\newtheorem{lemma}[theorem]{Lemma}    

\renewcommand{\qed}{\hfill{$\rule{6pt}{6pt}$}} 
\renewenvironment{proof}{\noindent{\bf Proof:}}{\qed\\}

\numberwithin{equation}{section}

\newcommand{\complex}{{\mathbb C}}
\newcommand{\reals}{{\mathbb R}}

\newcommand{\integers}{{\mathbb Z}}
\newcommand{\naturals}{{\mathbb N}}

\newcommand{\size}[1]{\left| #1 \right|}
\newcommand{\set}[1]{\left\{ #1 \right\}}
\newcommand{\floor}[1]{{\lfloor #1 \rfloor}}

\newcommand{\abs}[1]{\left| #1 \right|}

\DeclareMathOperator{\trace}{Tr}

\DeclareMathOperator{\order}{o}
\DeclareMathOperator{\Span}{span}
\DeclareMathOperator{\expct}{{\mathbb E}}

\DeclareMathOperator{\e}{e}

\newcommand{\id}{{\mathbb I}}
\newcommand{\allone}{{\mathbb J}}

\newcommand{\br}[1]{\mleft(#1\mright)}

\newcommand{\linear}{{\mathsf L}}

\newcommand{\qstate}{{\mathsf D}}

\newcommand{\ket}[1]{| #1 \rangle}
\newcommand{\bra}[1]{ \langle #1 |}
\newcommand{\ketbra}[2]{| #1 \rangle\!\langle #2 |}
\newcommand{\braket}[2]{\langle #1 | #2 \rangle }
\newcommand{\density}[1]{\ketbra{#1}{#1}}

\newcommand{\eqdef}{\coloneqq}
\newcommand{\tensor}{\otimes}

\newcommand{\intersect}{\cap}

\newcommand{\adjoint}{*}

\newcommand{\suppress}[1]{}

\newcommand{\etal}{\emph{et al.\/}}

\DeclareMathOperator{\entropy}{H}    
\DeclareMathOperator{\bentropy}{h}    
\DeclareMathOperator{\vnentropy}{S}    
\DeclareMathOperator{\mi}{I}    

\DeclareMathOperator{\VC}{VC\mbox{-}dim}
\DeclareMathOperator{\QPEX}{QPEX}
\DeclareMathOperator{\QAEX}{QAEX}
\DeclareMathOperator{\err}{err}
\DeclareMathOperator{\opt}{opt}
\DeclareMathOperator{\parity}{ps}

\newcommand{\sB}{{\mathsf B}}

\newcommand{\scA}{{\mathscr A}}
\newcommand{\scC}{{\mathscr C}}

\newcommand{\cA}{{\mathcal A}}
\newcommand{\cB}{{\mathcal B}}
\newcommand{\cE}{{\mathcal E}}
\newcommand{\cF}{{\mathcal F}}
\newcommand{\cX}{{\mathcal X}}
\newcommand{\cY}{{\mathcal Y}}

\newcommand{\tD}{{\widetilde D}}

\DeclareMathOperator{\hc}{\upalpha_{\mathrm{HC}}}    
\DeclareMathOperator{\dist}{\upgamma}    

\newcommand{\tW}{{\widetilde W}}

\title{\textbf{Optimal lower bounds for Quantum Learning via Information Theory}~\thanks{A preliminary version of the results in Section~\ref{sec-learning-functions} was included in the S.B.H.'s PhD thesis at University of Waterloo, December~2020~\cite{B20-phd-thesis}. 
An extended abstract of the results in Section~\ref{sec-coupon-collector} was included in the P.S.' bachelor's project report at Indian Institute of Science, April~2022.}
}

\author{
Shima Bab Hadiashar~\thanks{Zapata Computing Canada Inc., 25 Adelaide St.\ East, Toronto, ON, M5C 3A1, Canada. 
Most of this work was done while this author was with the 
Department of Combinatorics and Optimization,
and Institute for Quantum Computing, 
University of Waterloo. 
Email: \texttt{sbabhadi@uwaterloo.ca}~.
} \\
Zapata Computing Canada Inc.
\and
Ashwin Nayak~\thanks{Department of Combinatorics and Optimization,
and Institute for Quantum Computing, University
of Waterloo, 200 University Ave.\ W., Waterloo, ON,
N2L~3G1, Canada.
Email: \texttt{academic@ashwinnayak.info}~.
} \\
University of Waterloo
\and
Pulkit Sinha~\thanks{School of Computer Science,
and Institute for Quantum Computing, University
of Waterloo, 200 University Ave.\ W., Waterloo, ON,
N2L~3G1, Canada.
Most of this work was done during an undergraduate research project at University of Waterloo in Fall~2021, while this author was with the 
Department of Mathematics,
Indian Institute of Science.
Email: \texttt{psinha@uwaterloo.ca}~.
} \\
University of Waterloo}

\date{October 11, 2023}

\begin{document}

\maketitle

\begin{abstract}
Although a concept class may be learnt more efficiently using quantum samples as compared with classical samples in certain scenarios, Arunachalam and de Wolf~\cite{AW18-optimal-sample} proved that quantum learners are asymptotically no more efficient than classical ones in the quantum PAC and Agnostic learning models. They established lower bounds on sample complexity via quantum state identification and Fourier analysis. In this paper, we derive optimal lower bounds for quantum sample complexity in both the PAC and agnostic models via an information-theoretic approach. The proofs are arguably simpler, and the same ideas can potentially be used to derive optimal bounds for other problems in quantum learning theory.

We then turn to a quantum analogue of the Coupon Collector problem, a classic problem from probability theory also of importance in the study of PAC learning. Arunachalam, Belovs, Childs, Kothari, Rosmanis, and de Wolf~\cite{ABCKRW20-quantum-coupon-collector} characterized the quantum sample complexity of this problem up to constant factors. First, we show that the information-theoretic approach mentioned above provably does not yield the optimal lower bound, at least in its current form. As a by-product, we get a natural ensemble of pure states in arbitrarily high dimensions which are not easily (simultaneously) distinguishable, whereas the ensemble has close to maximal Holevo information. Second, we discover that the information-theoretic approach yields an asymptotically optimal bound for an approximation variant of the problem. Finally, we derive a sharper lower bound for the Quantum Coupon Collector problem, 
via the generalised Holevo-Curlander bounds on the distinguishability of an ensemble. All the aspects of the Quantum Coupon Collector problem we study rest on properties of the spectrum of the associated Gram matrix, which may be of independent interest.
\end{abstract}

\section{Introduction}
\label{sec-introduction}

Suppose we wish to predict a feature, given an individual from a
certain population, drawn from an unknown distribution. For example, we
may wish to recognize whether or not the individual has ``medium
build''.
We are given access to ``labelled'' random samples from the
population, i.e., we are told whether a given sample has the feature.
How many samples are needed for us to be able to predict the feature
with high probability? Valiant~\cite{val84-PAC} developed the
\emph{probably approximately correct\/} (PAC) model to formalize such
learning tasks. Later, Bshouty and Jackson~\cite{BJ99-DNF} extended the
notion of PAC learning to the quantum setting.
In reality, the labeled examples may be noisy, or they may not
necessarily be consistent with an underlying feature.
A more realistic model, the~\emph{agnostic\/} learning model, was introduced by
Haussler~\cite{Haus92-generalize-PAC}, and Kearns, Schapire and
Sellie~\cite{KSS94-agnostic}. In the agnostic model, the goal is to
make predictions that approximate the feature that is most consistent
with the examples. Like the PAC model, the agnostic model can
be extended to the quantum setting. Other models for quantum learning and their advantages, or lack thereof, have also been studied (see, e.g., Refs.~\cite{BPP21-generalization,CHCSSCC22-generalization}). A comprehensive discussion of these models and the broader area is beyond the scope of this work.

The learning models mentioned above are all defined formally in Section~\ref{sec-preliminaries}. Features
correspond to Boolean functions called \emph{concepts\/}, and a
collection of features is called a \emph{concept class\/}.  The goal in
these models is to find \emph{efficient\/} learners. A major figure of merit is
the \emph{sample complexity\/}, the minimum number of examples required
to learn a concept. The sample complexity of learning a concept class~$\scC$
in the above models depends heavily on a combinatorial quantity called the \emph{VC
dimension\/} of~$\scC$, in addition to an approximation
parameter~$\epsilon \in (0,1)$, and the probability of failure~$\delta
\in (0,1]$ of the learner.

In a series of works~\cite{BEHW89-Learn-VC,Han16-opt-PAC}, it has been
shown that the~$(\epsilon,\delta)$-PAC classical sample
complexity of a concept class with VC dimension~$d$ is
\begin{equation}
\label{eq:PAC-sample-complexity}
    \Theta \left( \frac{d}{\epsilon} + \frac{\log (1/\delta)}{\epsilon}\right) \enspace.
\end{equation}
For the agnostic model, the~$(\epsilon,\delta)$-agnostic classical
sample complexity is
\begin{equation}
\label{eq:agn-sample-complexity}
    \Theta \left( \frac{d}{\epsilon^2} + \frac{\log (1/\delta)}{\epsilon^2}\right) \enspace.
\end{equation}
The lower bound was shown by Vapnik and Chervonenkis~\cite{VC74-Agn-lb} and
this bound was shown to be achievable by Talagrand~\cite{Tal94-Agn-ub}.

Although a concept class may be learnt more efficiently using quantum samples as compared with classical samples in certain scenarios (see, e.g.,
Refs.~\cite{BJ99-DNF, SG04-sep-Q-C-learning}), the quantum sample
complexity in the PAC and agnostic models can be at most a
constant factor lower than their classical counterparts. In particular,
Eq.~\eqref{eq:PAC-sample-complexity} and
Eq.~\eqref{eq:agn-sample-complexity} also hold
for~$(\epsilon,\delta)$-PAC \emph{quantum\/} sample complexity
and~$(\epsilon, \delta)$-agnostic \emph{quantum\/} sample complexity,
respectively. The upper bounds carry over directly from the classical
learning algorithms, since a classical labelled example can be derived by
measuring a quantum example in the standard basis. The quantum lower
bounds are due to Arunachalam and de Wolf~\cite{AW18-optimal-sample}.

Arunachalam and de Wolf begin with a simple information-theoretic argument to
re-derive the \emph{classical\/} lower bounds in Eq.\eqref{eq:PAC-sample-complexity} and Eq.\eqref{eq:agn-sample-complexity}. They extend the argument to derive bounds for quantum
sample complexity, but the bounds are smaller by  a factor of~$\log(d/\epsilon)$.
They state that ``the logarithmic loss $\dotsc$ is actually inherent in
this information-theoretic argument'', and argue why this
is the case~\cite[Section~1.2.1]{AW18-optimal-sample}.
They present more sophisticated proofs via the ``pretty good measurement'' for quantum state identification and Fourier analysis for the optimal lower bounds.

From the arguments made by Arunachalam and de Wolf, one might be led to
believe that a logarithmic loss is inherent in the information-theoretic
approach. In this paper, we derive \emph{optimal\/} lower bounds for quantum sample complexity
in both the PAC and agnostic models via an information-theoretic approach.
The proofs are arguably simpler, and the same ideas can potentially be used to derive optimal
bounds for other problems in quantum learning theory.

The proofs we develop refine the information-theoretic argument given in
Ref.~\cite{AW18-optimal-sample}. The approach consists of
three steps. For concreteness, consider the PAC model. In the first
step, using the properties of a PAC learner, they show that the information contained in the quantum samples about the target concept is~$\Omega(d)$.  In the second step, they use
\emph{subadditivity\/} to bound this information by~$t$ times the
entropy of a single quantum sample (averaged over concepts), where~$t$
is the number of samples used by the learner. By bounding the entropy of
a single quantum sample, they conclude an~$\Omega \big(\tfrac{d}{\epsilon
\log(d/\epsilon)} \big)$ bound.  It turns out that the second step, i.e.,
the use of subadditivity, is the origin of the sub-optimality. Although their bound on information is tight for small values of~$t$, we point out that the information does not exceed~$d$, even as~$t$ goes to infinity. This hints at the possibility of a bound of order~$\epsilon t$ for sufficiently large~$t$. This is essentially what we establish.

We estimate the (mutual) information directly. Since the quantum examples are pure states, mutual information simplifies to the von Neumann entropy of a uniformly random example. We determine the spectral decomposition of the corresponding quantum state. The eigenvalues of the state involve combinatorial quantities that do not seem easy to analyse. Nonetheless, we show that the aggregate expression for entropy may be analysed using standard tail inequalities from probability theory. The resulting bound on entropy is~$\order(d)$, when~$t \in \order(d/\epsilon)$. This gives us the optimal lower bound on sample complexity. The arguments for agnostic learning go along the same lines. (On a closer look, Arunachalam and de Wolf encounter similar quantities in their approach, and it is likely that the ideas we use also lead to a simplification of the proofs they give.)

We then turn to the \emph{Coupon Collector\/} problem, a classic problem from probability theory which turns out to be important in the study of PAC learning. In this problem, given integers~$k,n$ (with~$1 < k < n$) the task is to independently draw uniformly distributed elements from an unknown size-$k$ subset~$S$ of~$[n]$, until all~$k$ elements of~$S$ have been observed. It is well-known that~$k \ln k + \Theta(k)$ samples are necessary and sufficient for the task to be accomplished with constant probability of success~\cite[Theorem~5.13]{MU17-probability-computing}. Viewed as a learning problem, the task is to learn the unknown set~$S$. It turns out that with~$k \eqdef n - 1$ the problem gives us an example of a concept class that is~$(1/n, 1/4)$-PAC learnable with~$\Theta(n)$ samples, albeit \emph{improperly\/}, whereas \emph{proper\/} PAC learning with the same parameters (i.e., producing a size-$(n - 1)$ subset as the hypothesis satisfying the PAC learning condition) requires~$\Theta(n \ln n)$ samples. We refer the reader to, e.g., Ref.~\cite[Section~5]{ABCKRW20-quantum-coupon-collector} for the details.

Motivated by the question of proper versus improper \emph{quantum\/} learning, Arunachalam, Belovs, Childs, Kothari, Rosmanis, and de Wolf~\cite{ABCKRW20-quantum-coupon-collector} studied a quantum analogue of the Coupon Collector problem. In the Quantum Coupon Collector problem, the learner has access to quantum samples, i.e., uniform superpositions over the elements of an unknown size-$k$ subset~$S$ of~$[n]$. The task is to learn~$S$ with high probability.
Arunachalam \etal\ showed that the quantum sample complexity of the problem matches that of classical coupon collection until~$k$ gets ``too close'' to~$n$. More precisely, with~$m \eqdef n - k$, they proved that the sample complexity of the Quantum Coupon Collector problem with constant probability of error~$< 1$ is
\begin{equation}
\label{eq-qcc}
\Theta(k \log \min \set{k, m + 1}) \enspace.
\end{equation}
In particular, when~$k = n - 1$, the sample complexity is~$\Theta(n)$. Thus, unlike in the classical case, the quantum version of the problem \emph{does not\/} separate proper from improper quantum PAC learning. 

Arunachalam \etal\ proved the lower bound for the Quantum Coupon Collector problem using the \emph{adversary bound\/} (see, e.g., Ref.~\cite{B12-PhD-thesis}) and properties of the Johnson association scheme. We study potentially simpler information-theoretic techniques for analysing the problem. First, we show that the natural approach we follow in Section~\ref{sec-learning-functions} \emph{provably\/} does not yield the optimal lower bound, at least in its current form (\Cref{sec-qcc-it}). As a by-product, we get a natural ensemble of pure states in arbitrarily high dimensions which are not easily (simultaneously) distinguishable, whereas the ensemble has close to maximal Holevo information. (Since the states in the ensemble are pure, the Holevo information equals the von Neumann entropy of the ensemble average.) This complements early work in quantum information theory due to Jozsa and Schlienz~\cite{JS00-distinguishability}, which presents ensembles of states in three dimensions which can be ``deformed'' so that the von Neumann entropy increases whereas the states become \emph{pairwise\/} less distinguishable. 

Although it fails for the Quantum Coupon Collector problem, we discover that the standard information-theoretic approach yields an asymptotically optimal bound for a possibly harder variant of the problem, namely that of \emph{approximating\/} the set to within a constant number of incorrect elements in expectation (\Cref{thm-qcc-mismatches}). Finally, we take a different route---via the generalised Holevo-Curlander bounds~\cite{Holevo79-distinguishability,Curlander79-distinguishability,ON99-converse-channel-coding,Tyson09-distinguishability} on the distinguishability of an ensemble of states---to derive 
a sharper lower bound for the Quantum Coupon Collector problem (\Cref{cor-qcc-lb}). In fact, we establish a bound of~$(1 - \order_k(1)) k \ln \min \set{k, m+1}$, which has the optimal leading order term for~$k \le n/2$, including the constant factor~$1$. We conjecture that this optimality holds for~$k > n/2$ as well.

All the aspects of the Quantum Coupon Collector problem we study rest on properties of the spectrum of the associated Gram matrix. In particular, we make essential use of a recurrence relation involving the spectrum of the Gram matrix. Similar to the analysis in Ref.~\cite{ABCKRW20-quantum-coupon-collector}, the recurrence allows us to interpret certain spectral quantities as probabilities arising from random walks that resemble the classical Coupon Collector process. In contrast with Ref.~\cite{ABCKRW20-quantum-coupon-collector}, though, the spectral quantities of relevance turn out to be the product of the eigenvalues with the corresponding multiplicity (informally speaking). The corresponding probabilities pertain to different events, and call for fresh analysis. The relationship with the aforesaid random walks lies at the heart of the delicate bounds on the spectral quantities, the corresponding entropy, etc., that are required for the results. These properties of the spectrum, as also the ideas underlying the bounds for PAC and Agnostic learning may be of wider applicability and interest.

\paragraph{Organization of the paper.}
We present basic elements of classical and quantum information theory, and the formal definition of the models and problems in Learning Theory studied in this article in \Cref{sec-preliminaries}. We study sample complexity in the quantum PAC and Agnostic learning models in \Cref{sec-learning-functions}. We describe our approach in detail for the quantum PAC model (\Cref{sec-pac-learning}) and illustrate the versatility of the approach in \Cref{sec-agnostic-learning}, by applying it to the quantum Agnostic model.  

\Cref{sec-coupon-collector} is devoted to the Quantum Coupon Collector problem. We develop properties of the spectrum of the Gram matrix associated with the problem in \Cref{sec-spectrum}. We demonstrate the failure of the standard information-theoretic argument in \Cref{sec-qcc-it}, derive the lower bound for approximation in \Cref{sec-mismatches}, and wrap up with the 
lower bound for the Quantum Coupon Collector problem in \Cref{sec-coupon-collection-lb}.

\paragraph{Acknowledgements.}
We thank Marco Tomamichel for pointing out an error in an earlier version of the proof of Theorem~\ref{thm-PAC-learning}, and Jon Tyson for clarifying the history of the generalised Holevo-Curlander bounds. This research was supported in part by NSERC Canada and a grant from Fujitsu Labs America. S.B.H.\ was also supported by an Ontario Graduate Scholarship. P.S.\ is also supported by a Mike and Ophelia Lazaridis Fellowship.

\section{Preliminaries}
\label{sec-preliminaries}

\paragraph{Some notation.}

For a positive integer~$d$, we denote the set~$\{0,1,\ldots,d\}$ as~$\naturals_{d+1}$ and the set~$\{1,\ldots,d\}$ as~$[d]$. For a string~$a \in \set{0,1}^d$ and~$i \in \naturals_{d+1}$ we use the convention that~$a_i$ is~$0$ if~$i = 0$, and~$a_i$ is the~$i$-th bit of~$a$ otherwise. For~$u \in \naturals_{d+1}^t$ for some~$t \ge 1$, we denote the string~$a_{u_1} a_{u_2} \dotsb a_{u_t}$ by~$a_u$. Let~$r \ge 1$ be a positive integer. The \emph{Hamming weight\/} of a string~$u \in \naturals_{d+1}^{r}$ is the number of non-zero symbols in the string and we denote it as~$|u|$. For~$u \in \naturals_{d+1}^{r}$ and~$x \in \set{0,1}^{r}$, we use~$u_x$ to denote the string~$u_{i_1} u_{i_2} \dotsb u_{i_j}$ where~$i_1, i_2, \dotsc, i_j$ are all the positions~$i$ where~$x_i = 1$, in increasing order. I.e., $u_x$ is ``projection'' of~$u$ onto the subset of indices in~$[r]$ whose characteristic vector is~$x$. We define the \emph{parity signature\/} of a string~$u \in \naturals_{d+1}^{r}$ as the function~$\parity : \naturals_{d+1}^{r} \rightarrow \{0,1\}^{d}$ such that for~$i \in [d]$, the~$i$-th bit of~$\parity(u)$ is the parity of the number of times~$i$ occurs in the string~$u$, i.e., the number of occurrences of~$i$ in~$u$, modulo~$2$. As an example, consider~$d=r=4$ 
and~$a = (0,1,1,2)$. Then,~$\parity(a) = (0,1,0,0)$. We denote the number of strings in the 
set~$[d]^{r}$ with parity signature~$b \in \{0,1\}^{d}$ as~$n_{r,b}$. Note that~$n_{r,b} = n_{r,b'}$ whenever~$|b| = \size{b'}$. For an integer~$h \in \naturals_{d+1}$, we also use~$n_{r,h}$ for the number of strings in~$[d]^{r}$ that have a specific parity signature, say~$b$, with Hamming weight~$\size{b} = h$.

Where the base of the logarithm is not specified, we take the base to be~$2$.

\paragraph{Classical information theory.}

We briefly describe the basics from information theory that we use in this work, and refer the reader to the text~\cite{CT06-info-theory} for a thorough introduction to the subject.

We only consider random variables with finite sample spaces, and denote them with capital letters, like~$X$,~$Y$ and~$Z$. For a random
variable~$X$ over sample space~$\cX$, the \emph{Shannon entropy\/} of~$X$ is defined as
\begin{equation*}
    \entropy(X) \quad \coloneqq \quad \sum_{x \in \cX} p(x) \log \frac{1}{p(x)} \enspace,
\end{equation*}
where~$p$ is the distribution of~$X$. For jointly distributed random variables~$X$ and~$Y$, the \emph{conditional entropy} of~$X$ given~$Y$ is defined as
\begin{equation*}
    \entropy(X | \, Y) \quad \coloneqq \quad \sum_{y \in \cY} \Pr[Y=y] \entropy(X|Y=y) \enspace,
\end{equation*}
where~$\cY$ is the sample space of~$Y$.
By definition, conditional (Shannon) entropy is a non-negative quantity and~$\entropy(X | \, Y) \le \entropy(X)$. Moreover, Shannon entropy satisfies the following \emph{chain rule} :
\begin{equation*}
    \entropy(XY) \quad = \quad \entropy(X) + \entropy(Y| \, X) \enspace.
\end{equation*}
Suppose~$X$ and~$Y$ are two random variables with the same sample space such that~$\Pr[X \neq Y] \leq \epsilon$ for some~$\epsilon \in [0,1]$. The \emph{Fano inequality\/} bounds the conditional entropy of~$X$given~$Y$ as
\begin{equation*}
    \entropy(X| \, Y) \quad \le \quad \bentropy(\epsilon) + \epsilon \log(|\cX|-1) \enspace,
\end{equation*}
where~$\cX$ is the sample space of~$X$ and has cardinality~$|\cX|$,  and~$\bentropy(\epsilon) \coloneqq - \epsilon \log \epsilon - (1-\epsilon) \log (1-\epsilon)$ denotes the binary entropy function.

\paragraph{Tail bounds.}

Let~$X_1, X_2, \ldots, X_n$ be~$n$ independent random variables over~$\set{0,1}$ and~$X$ denote their sum, i.e.,~$X \eqdef \sum_{i=1}^{n} X_i$. Let~$\mu$ be the expected value of~$X$. The \emph{Chernoff bound\/}~\cite{MR95-randomized-algorithms} states that~$X$ is concentrated around~$\mu$. In particular, for~$\delta \ge 0$, we have
\begin{equation}
\label{eq:Chernoff-bound}
    \Pr \left[ X \ge (1+\delta) \mu \right] \quad \leq \quad \exp\left(- \, \frac{\delta^2 \mu}{2+\delta} \right) \enspace.
\end{equation}
The above bound is also known as the \emph{multiplicative\/} Chernoff bound in the literature. 
In the case that~$X_i$ are all Bernoulli random variables with~$\expct X_i = 1/2$, the random variable~$X$ has the binomial distribution~$\sB(n, 1/2)$ with~$\mu = n/2$, and we have a tighter bound. For~$k \le n/2$,
\begin{equation}
\label{eq:concent-binom-1/2}
    \Pr \left[ X \le k\right] \quad = \quad \Pr \left[ X \ge n-k\right] \quad = \quad \frac{1}{2^n} \sum_{r=0}^{k} \binom{n}{r} \quad \leq \quad 2^{-\left(1-\bentropy(k/n)\right) n} \enspace.
\end{equation}
See, e.g., Ref.~\cite[Lemma 16.19, page 427]{FG06-PCT} for this bound.
 
\paragraph{Quantum Information Theory.}

For a thorough introduction to basics of quantum information, we refer the reader to the book by Watrous~\cite{W18-TQI}. We briefly review the notation and some results that we use in this article. 

We denote (finite dimensional) Hilbert spaces either by capital script letters like~$\cA$ and $\cB$, or directly as~$\complex^m$ for a positive integer~$m$.
A \emph{register\/} is a physical quantum system, and we denote it by a capital letter, like~$A$ or~$B$. A quantum register~$A$ is associated with a Hilbert space~$\cA$, and the state of the register is specified by a unit-trace positive semi-definite operator on~$\cA$. The state is called a \emph{quantum state\/}, or also as a \emph{density operator\/}. We denote
quantum states by lower case Greek letters, like~$\rho$ and~$\sigma$.  We use notation such as~$\rho^A$ to indicate that register~$A$ is in state~$\rho$, and 
may omit the superscript when the register is clear from the context. We use \emph{ket\/} and \emph{bra\/} notation to denote unit vectors 
and their adjoints in a Hilbert space, respectively. A quantum state is \emph{pure\/} if it has rank one, i.e.,~$\rho = \density{\phi}$ where~$\ket{\phi}$ is a unit vector. 
Any register~$A$ with Hilbert space~$\complex^S$, for some non-empty finite set~$S$, is called \emph{classical\/} if any joint state with any other register~$B$ is invariant under the operator
\[
\sigma^{AB} \quad \mapsto \quad \sum_{x \in S} \big( \density{x}^A \tensor \id^B \big) \, \sigma \, \big( \density{x}^A \tensor \id^B \big) \enspace,
\]
where~$\set{\ket{x}: x \in S}$ is the canonical basis for~$\complex^S$. In other words, the joint state of~$A$ and any other register~$B$ is always of the form~$\sum_x p_x \density{x}^A \tensor \rho_x^B$, where~$(p_x)$ is some distribution over~$S$. A classical register corresponds to a random variable whose probability distribution is determined by the diagonal entries of the state of the register.  

We denote the set of all linear operators on Hilbert space~$\cA$
by~$\linear(\cA)$, and we denote quantum channels, i.e., completely positive and trace-preserving linear maps from the space of linear operators
on a Hilbert space to another such space, by capital Greek letters
like~$\Psi$.

For a register~$A$ with quantum state~$\rho$, the \emph{von Neumann entropy\/}~$\vnentropy(A)_{\rho}$ of~$A$ is defined as
\begin{equation*}
    \vnentropy(A)_{\rho} \quad \coloneqq \quad - \trace \left( \rho \log \rho \right) \enspace.
\end{equation*}
This definition coincides with the Shannon entropy of the spectrum of~$\rho^{A}$. Hence, the entropy of a quantum state is invariant under the action of unitary transformations. For registers~$A$ and~$B$ with joint quantum state~$\rho^{AB}$, the \emph{mutual Information\/}~$\mi(A:B)_{\rho}$ of~$A$ and~$B$ is defined as
\begin{equation*}
    \mi(A:B)_{\rho} \quad \coloneqq \quad \vnentropy(A)_{\rho} + \vnentropy(B)_{\rho} - \vnentropy(AB)_{\rho} \enspace.
\end{equation*}
In the above notation, the subscript~$\rho$ may be omitted when the state is clear from the context. The mutual information of~$A$ and~$B$ is monotonic under the application of local quantum channels.
In particular, for a quantum channel~$\Psi : \linear(\cB) \rightarrow \linear(\cB')$ mapping states of register~$B$ to those of register~$B'$, we have
\begin{equation*}
    \mi(A:B)_{\rho} \quad \geq \quad \mi(A:B')_{(\id \tensor \Psi)(\rho)} \enspace.
\end{equation*}
This is also known as the \emph{Data Processing Inequality\/}.

\paragraph{Quantum learning theory.}

We refer the readers to the book~\cite{SB14-ML} for an introduction to machine learning theory and the survey~\cite{AdW17-survey} for an introduction to  quantum learning theory. Here,  we briefly review the
notation related to the PAC model and agnostic model that we study in this article.

We refer to a Boolean function~$c: \set{0,1}^n \rightarrow \set{0,1}$ as a \emph{concept}. 
We may also think of a concept as a bit-string in~$\set{0,1}^N$ for~$N\coloneqq 2^n$ which contains 
the value of~$c$ at all possible~$n$-bit strings. A \emph{concept class\/} is a 
subset~$\scC \subseteq \set{0,1}^N$ of Boolean functions. For a concept~$c$, we refer to~$c(x)$ 
as the \emph{label\/} of~$x \in \set{0,1}^n$, and the tuple~$(x,c(x))$ as 
a \emph{labeled example}. We say a concept class~$\scC$ is \emph{non-trivial\/} if it contains two distinct concepts~$c_1,c_2$ such that for some~$x_1, x_2 \in \set{0,1}^n$, we have~$c_1(x_1) = c_2(x_1)$ and~$c_1(x_2) \ne c_2(x_2)$. 

A crucial combinatorial quantity in learning Boolean functions is the \emph{VC dimension\/} of a concept class, introduced by Vapnik and Chervonenkis~\cite{VC71-VC-dim}. We say a set~$S \eqdef \set{s_1,\ldots,s_d} \subseteq \set{0,1}^n$ is \emph{shattered} by a concept class~$\scC$ if for every~$a \in \set{0,1}^{d}$, there exists a concept~$c \in \scC$ such that~$(c(s_1),\ldots,c(s_d)) = a$. The \emph{VC dimension\/} of~$\scC$, denoted as~$\VC(\scC)$, is the size of the largest set shattered by~$\scC$.  

\paragraph{PAC model.}

Consider a concept class~$\scC \subseteq \set{0,1}^{N}$. The PAC (\emph{probably approximately correct\/}) model for learning concepts was introduced---in the classical setting---by Valiant~\cite{val84-PAC}, and was extended to the quantum setting by Bshouty and Jackson~\cite{BJ99-DNF}. In the quantum PAC model, a learning algorithm is given a~\emph{quantum PAC example oracle\/}~$\QPEX(c,D)$ for an unknown concept~$c \in \scC$ and an 
unknown distribution~$D$ over~$ \set{0,1}^n$. The oracle~$\QPEX(c,D)$ does not have any inputs. When invoked, it outputs a superposition of labeled examples of~$c$ with amplitudes given by the distribution~$D$, 
namely, the pure state
\begin{equation*}
    \sum_{x \in \set{0,1}^n} \sqrt{D(x)} \, \ket{x,c(x)} \enspace,
\end{equation*}
which we call a \emph{quantum sample\/}, or simply a \emph{sample\/}. Note that measuring a quantum sample in the computational basis gives us a labeled example distributed according to~$D$, i.e., a classical sample.
We say a Boolean function~$h$, commonly called a \emph{hypothesis\/}, is an~$\epsilon$-\emph{approximation\/} of~$c$ (or has error~$\epsilon$) with respect to distribution~$D$, if 
\begin{equation}
     \Pr_{x \sim D} \left[ h(x) \neq c(x) \right] \quad \leq \quad \epsilon   \enspace.
\end{equation}
Given access to the oracle~$\QPEX(c,D)$, the goal of a quantum PAC learner is to find a hypothesis~$h$ that is an~$\epsilon$-approximation of~$c$ with sufficiently high success probability. 
\begin{definition}
For~$\epsilon , \delta \in [0,1]$, we say that an algorithm~$\scA$ is an~$(\epsilon,\delta)$-PAC quantum learner
for concept class~$\scC$ if for every~$c\in\scC$ and distribution~$D$, given access to~$\QPEX(c,D)$,
 with probability at least~$1-\delta$,~$\scA$ outputs a hypothesis~$h\in \set{0,1}^N$ which is an~$\epsilon$-approximation of~$c$
\end{definition}
The \emph{sample complexity\/} of a quantum learner~$\scA$ is the maximum number of times~$\scA$ invokes the oracle~$\QPEX(c,D)$ for any concept~$c \in \scC$ and any distribution~$D$ over~$\set{0,1}^n$. The~$(\epsilon,\delta)$-\emph{PAC quantum sample complexity\/} of a concept class~$\scC$ is the minimum sample complexity of a~$(\epsilon,\delta)$-PAC quantum learner for~$\scC$. Since we can readily derive classical samples from quantum ones, quantum learning algorithms are at least as efficient as classical ones in terms of sample complexity, as well as other measures such as time and space complexity.

\paragraph{Agnostic model.}

In the PAC model it is assumed that examples are generated perfectly according to some unknown concept~$c \in \scC$. The \emph{agnostic\/} model is a more realistic model where examples correspond to random labeled pairs~$(x,l)$ distributed according to an unknown distribution~$D$ over~$\set{0,1}^{n+1}$, not necessarily derived from a specific concept~$c \in \scC$. This model was introduced in the classical setting by Haussler~\cite{Haus92-generalize-PAC}, and Kearn, Schapire, and Sellie~\cite{KSS94-agnostic}. It was first studied in the quantum setting by Arunachalam and de Wolf~\cite{AW18-optimal-sample}. 
In the agnostic model, the learning algorithm has access to an agnostic quantum oracle~$\QAEX(D)$ for an unknown distribution~$D$ over~$\set{0,1}^{n+1}$. When invoked, the oracle~$\QAEX(D)$ outputs the superposition
\begin{equation*}
    \sum_{(x,l) \in \set{0,1}^{n+1}} \sqrt{D(x,l)} \, \ket{x,l} \enspace.
\end{equation*}
The \emph{error\/} of a hypothesis~$h \in \set{0,1}^N$ under distribution~$D$ is defined as
\begin{equation}
    \err_{D}(h) \quad \coloneqq \quad \Pr_{(x,l)\sim D} \left[ h(x) \neq l \right] \enspace.
\end{equation}
For a concept class~$\scC$, the minimal error achievable is defined as
\begin{equation*}
    \opt_{D}(\scC) \quad \coloneqq \quad \min_{c \, \in \,\scC} \quad \err_{D}(c) \enspace.
\end{equation*}
Given access to a~$\QAEX(D)$ oracle, the goal of quantum agnostic learner is to find a hypothesis~$h \in \scC$ with error not ``much larger'' than~$\opt_{D}(\scC)$.
\begin{definition}
For~$\epsilon , \delta \in [0,1]$, we say a learning algorithm~$\scA$ is an~$(\epsilon,\delta)$-agnostic quantum learner
for~$\scC$ if for every distribution~$D$, given access to~$\QAEX(D)$,
 with probability at least~$1-\delta$,~$\scA$ outputs a hypothesis~$h\in \scC$ such that~$\err_{D}(h) \leq \opt_{D}(\scC) + \epsilon$.
\end{definition}
As in the PAC model, we define the \emph{sample complexity\/} of a quantum learner~$\scA$ as the maximum number of times~$\scA$ invokes the oracle~$\QAEX(D)$ for any distribution~$D$ over~$\set{0,1}^{n+1}$. The~$(\epsilon,\delta)$-\emph{agnostic quantum sample complexity\/} of a concept class~$\scC$ is the minimum sample complexity of an~$(\epsilon,\delta)$-agnostic quantum learner for~$\scC$. Again, since a classical labeled example may be generated from a quantum sample by measurement in the computational basis, quantum learning algorithms in this model are also at least as efficient as their classical counterparts (with respect to the standard complexity measures).

\paragraph{Quantum Coupon Collection.}

Let~$n$ be an integer~$\ge 3$. For a positive integer~$k \in (1,n)$, let~$S$ be a~$k$-element subset of~$[n]$, and let~$\ket{\psi_S}$ denote the uniform superposition over the elements of~$S$:
\begin{equation}
\label{eq-qsample}
 \ket{\psi_S} \quad \eqdef \quad \frac{1}{\sqrt k} \sum_{i \in S} \ket{i} \enspace.   
\end{equation}
This is a quantum analogue of a uniformly random sample from~$S$, and we call this a \emph{quantum sample\/} of~$S$. For ease of notation, we define~$m \eqdef n - k$. 

In the Quantum Coupon Collector problem, we are given~$n,k$ and quantum samples of an arbitrary but fixed, unknown~$k$-element subset~$S$, and we would like to learn the subset using as few samples as possible. By ``learning the subset'', we mean that we would like to determine, with probability at least~$1 - \delta$ for some parameter~$\delta \in [0, 1)$, all the~$k$ elements of the set~$S$. We are interested in the quantum sample complexity of the problem, i.e., the least number of samples required by a quantum algorithm to learn the set with probability of error at most~$\delta$. Observe that by permuting the elements of~$[n]$ by a uniformly random permutation, we can show that the optimal worst-case error equals the optimal average-case error under the uniform distribution over the sets.

We also study a variant of this problem, in which given a parameter~$\Delta \ge 0$ the goal of the learning algorithm is to \emph{approximate\/} the set~$S$ to within~$\Delta$ in expectation, i.e., output a size-$k$ subset~$S'$ such that the expected number of \emph{mismatches\/}~$\size{S \setminus S'}$ is bounded: $\expct \size{S \setminus S'} \le \Delta$. Note that when $\Delta \le 1$, with probability at least~$1 - \Delta$, such an algorithm outputs the set~$S$. So the condition that the expected number of mismatches is bounded by~$1$ is a stronger condition than the learnability of the set~$S$.

\section{Learning Boolean functions}
\label{sec-learning-functions}

In this section, we give simple, information-theoretic lower bounds for the sample complexity of learning Boolean functions in the PAC and agnostic learning models.

\subsection{Sample complexity of PAC learning}
\label{sec-pac-learning}

The lower bound for PAC learning in Eq.~\eqref{eq:PAC-sample-complexity} contains two parts, one depends on VC dimension and the other one is independent of it. The  VC-independent part of the lower bound was shown by At{\i}c{\i} and Servedio~\cite{AS05-bounds-QL}. For completeness, we provide a proof here.

\begin{lemma}
\label{lemma-PAC-no-VC}
Let~$\scC$ be a non-trivial concept class. For every~$\delta \in (0,1/2)$ and~$\epsilon \in (0,1/4)$, an~$(\epsilon, \delta)$-PAC quantum learner for~$\scC$ has quantum sample complexity at least~$\Omega\br{\tfrac{1}{\epsilon} \log \tfrac{1}{\delta}}$.
\end{lemma}

\begin{proof}
Since~$\scC$ is non-trivial, there exist concepts~$c_1,c_2 \in \scC$ and inputs~$x_1,x_2 \in \set{0,1}^n$ such that~$c_1(x_1) = c_2(x_1)$ and~$c_1(x_2) \ne c_2(x_2)$. Consider the distribution~$D$ defined as follows:~$D(x_1) \eqdef 1-2\epsilon$ and~$D(x_2) \eqdef 2\epsilon$. Let~$\ket{\psi_i} \coloneqq \sqrt{1- 2\epsilon} \, \ket{x_1,c_i(x_1)} + \sqrt{ 2\epsilon} \, \ket{x_2,c_i(x_2)}$ for~$i \in \set{1,2}$. 

Under the distribution~$D$, no hypothesis can simultaneously~$\epsilon$-approximate~$c_1$ and~$c_2$. So, the hypotheses output with probability at least~$1 - \delta$ by an~$(\epsilon,\delta)$-PAC quantum learner for~$\scC$, given samples for~$c_1$ and~$c_2$ are different. The learner thus distinguishes the states~$\ket{\psi_1}^{\tensor t}$ and~$\ket{\psi_2}^{\tensor t}$ with success probability at least~$1-\delta$, where~$t$ is its sample complexity. 

On the other hand, by Holevo-Helstrom theorem (see, e.g., Ref.~\cite[Theorem 3.4]{W18-TQI}),~$\ket{\psi_1}^{\tensor t}$ and~$\ket{\psi_2}^{\tensor t}$ are distinguishable with probability at most
\[
\tfrac{1}{2}+\tfrac{1}{2} \sqrt{1-\abs{\braket{\psi_1}{\psi_2}}^{2t}} \enspace.
\]
So, we have~$\braket{\psi_1}{\psi_2}^{t} \le 2\sqrt{\delta (1-\delta)}$. Since~$\braket{\psi_1}{\psi_2} = 1- 2\epsilon$, it follows that~$t \in \Omega\br{\tfrac{1}{\epsilon}\log \tfrac{1}{\delta} }$.
\end{proof}

Next, we prove the VC-dependent part of the lower bound.
\begin{theorem}
\label{thm-PAC-learning}
Let~$n \ge 2$,~$d \ge 1$,~$N \eqdef 2^n$, and let~$\scC \subseteq \set{0,1}^N$ be a concept class with~$\VC(\scC)$ being~$d+1$. There is a universal constant~$d_0$ such that if~$d \ge d_0$,~$\epsilon \in (0,1/4)$ 
and~$\delta \in (0,1/8)$, every~$(\epsilon, \delta)$-PAC quantum learner for~$\scC$ has sample 
complexity~$\Omega (d/\epsilon)$.
\end{theorem} 

\begin{proof}
Suppose there is an~$(\epsilon, \delta)$-PAC quantum learner~$\scA$ for~$\scC$ 
with quantum sample complexity~$t$. 
Let $S \eqdef \{s_0,\ldots, s_d\} \subseteq \{0,1\}^n$ be 
a maximal set shattered by the concept class~$\scC$. For each~$a \in \{0,1\}^d$, let~$c_a \in \scC$ be a concept such that~$c_a(s_0) = 0$ and~$c_a(s_i) = a_i$ for 
all~$i \in [d]$. Let~$D$ be a probability distribution over~$\{0,1\}^n$ defined as~$D(s_0) \eqdef 1-4\epsilon$,~$D(s_i) \eqdef 4\epsilon/d$ for~$i \in [d]$, and~$D(x) \eqdef 0$ for~$x\in \set{0,1}^n \setminus S$. Since~$D$ is only supported over the set~$S$, in the rest of the proof, we write~$i$ instead of~$s_i$ for the sake of brevity. Hence, we consider~$D$ as a distribution over~$\naturals_{d+1}$ and write the PAC quantum example for~$c_a$ according to~$D$ as
\[
\ket{\psi_a} \quad \eqdef \quad \sum_{i \in \naturals_{d+1}} \sqrt{D(i)} \, \ket{i, a_i} 
    \quad = \quad \sqrt{1-4\epsilon} \; \ket{0,0} + \sqrt{ \frac{4 \epsilon}{ d} } \; \sum_{i=1}^{d} \ket{i,a_i} \enspace.
\]
Let~$\rho_a$ denote~$t$ copies of the example~$\density{\psi_a}$ and let~$\rho^{AB_1 \dotsb B_t}$ be the following classical-quantum state
\begin{equation}
    \rho^{AB_1 \dotsb B_t} \quad \eqdef \quad \frac{1}{2^d} \sum_{a \in \{0,1\}^d} \density{a}^{A} \tensor \rho^{B_1 \dotsb B_t}_{a} \quad 
    = \quad \frac{1}{2^d} \sum_{a \in \{0,1\}^d} \density{a} \tensor \density{\psi_a}^{\tensor t} \enspace.
\end{equation}
Let~$B \coloneqq B_1 \dotsb B_t$. First, we show that~$\mi(A:B)_{\rho} \in \Omega(d)$ using the hypothesis that~$\scC$, and therefore, its subset~$\{c_a: a \in \{0,1\}^d\}$, is~$(\epsilon, \delta)$-PAC quantum learnable using~$t$ quantum examples. Then, we show that if~$t \in \order(d/\epsilon)$, we have~$\mi(A:B)_{\rho} \in \order(d)$. This implies that~$t \in \Omega(d/\epsilon)$.

We begin with the lower bound on the mutual information.
\begin{lemma}
\label{lem-pac-I(A:B)-lb}
$ \mi(A : B) \quad \ge \quad (1 - \bentropy(3/8)) d \, $.
\end{lemma}
\begin{proof}
There are several ways of showing that~$\mi(A:B)_{\rho} \in \Omega(d)$; see the proof of Theorem~16 in Ref.~\cite{AW18-optimal-sample} for one. We give a different argument.

Suppose the learner~$\scA$ is given the~$t$ quantum examples in the 
register~$B$ and produces the (classical) output in register~$F$. Let~$E$ be the a binary random variable indicating whether or not~$\scA$ outputs a desirable hypothesis, i.e., a hypothesis~$f \in \set{0,1}^N$ 
satisfying~$\Pr_{i \sim D} \left[ f(i) \neq c_a(i) \right] \leq \epsilon$. The event~$E = 1$ occurs with probability at least~$1-\delta$. For any~$i \in [d]$,
\begin{IEEEeqnarray*}{rCl}
\Pr[F_i \neq A_i ] \quad 
    & = & \quad \Pr[ E = 0 ] \; \Pr[ F_i \neq A_i \,|\, E = 0] +  \Pr[ E = 1 ] \; \Pr[ F_i \neq A_i \,|\, E = 1] \\
    & \le & \quad \delta + \Pr[ F_i \neq A_i \,|\, E = 1] \enspace.
\end{IEEEeqnarray*}
From the PAC learning condition, we have
\[
\frac{4 \epsilon}{d} \sum_{i = 1}^d \Pr[ F_i \neq A_i \,|\, E = 1]
    \quad \le \quad \epsilon \enspace,
\]
whereby
\[
\frac{1}{d} \sum_{i = 1}^d \Pr[ F_i \neq A_i ]
    \quad \le \quad \delta + \frac{1}{4} \quad \le \quad \frac{3}{8} \enspace.
\]
By the Data Processing Inequality and the hardness of the Index function under the uniform distribution over inputs~\cite[Lemma~V.3]{KNTZ07-pointer-jumping}, we get
\begin{equation*}
\label{eq-pac-I(A:B)-lb}
\mi(A : B) \quad \ge \quad \mi(A : F) \quad \ge \quad (1 - \bentropy(3/8)) d \enspace,
\end{equation*}
as claimed.
\end{proof}

Note that~$t = 0$ is not possible. Assume that~$t = \nu d/\epsilon $ for some~$\nu \in (0,1)$. Otherwise, the bound stated in the theorem holds. We show that there are positive universal constants~$\beta, d_0$ such that if~$\nu < \beta$ and~$d \ge d_0$, the mutual information~$\mi(A:B)$ violates \Cref{lem-pac-I(A:B)-lb}.

Notice that~$\rho^{AB}$ and~$\rho^{A}$ have the same eigenvalues, as~$B$ is in a pure state for a fixed value~$a$ of~$A$. Hence,~$\mi(A:B)_{\rho} = \vnentropy(A)_{\rho} + \vnentropy(B)_{\rho} - \vnentropy(AB)_{\rho} = \vnentropy(B)_{\rho} \,$. Arunachalam and de Wolf~\cite[proof of Theorem~16]{AW18-optimal-sample} bound the entropy~$\vnentropy(B)_{\rho}$ by~$t \vnentropy(B_1)_{\rho} \,$, using subadditivity. Since~$\vnentropy(B_1)_{\rho}$ is of the order of~$\epsilon \log d$, they get a lower bound that is a factor~$\log d$ smaller than optimal. Although this bound is tight for small values of~$t$, note that the state~$\rho^B$ has rank at most~$2^d$. Hence, its entropy does not exceed~$d$, even as~$t$ goes to infinity. This hints at the possibility of a bound on entropy of order~$\epsilon t$ for sufficiently large~$t$. This is essentially what we establish, by a direct analysis of the spectrum of~$\rho^B$.

We extend the distribution~$D$ to~$\naturals_{d+1}^k$ for~$k > 1$ by defining~$D(w) \eqdef \prod_{i \in [k]} D(w_i)$, for any~$w \in \naturals_{d+1}^k \,$. We have
\begin{equation*}
    \rho^{B} \quad = \quad \frac{1}{2^d} \sum_{a \in \{0,1\}^d} 
    \density{\psi_a}^{\tensor t} \quad = \quad \frac{1}{2^d} 
    \sum_{a \in \{0,1\}^d} \ \sum_{u,v \in \naturals_{d+1}^{t}} 
    \sqrt{D(u) D(v)} \; \ketbra{u}{v}^{I} \tensor 
    \ketbra{a_{u}}{a_{v}}^{L} \enspace,
\end{equation*}
where we partition register~$B$  into registers~$I$ and~$L$ which contain the sequences of indices and labels from the~$t$ samples, respectively. 

By considering small~$t$ (e.g., $t = 1,2$), we see that~$\rho^{B}$ is diagonalized when we express register~$L$ in the Hadamard basis. This generalizes to all~$t$, and allows us to obtain an explicit spectral decomposition for the state. Recall from \Cref{sec-preliminaries} that~$n_{l,h}$ is the number of strings in~$ [d]^l$ that have a specific parity signature, say~$b$, with Hamming weight~$ |b| = h$.

\begin{lemma}
\label{lem-spectrum1}
The state~$\rho^{B}$ has~$\min \set{d+1, t+1}$ non-zero eigenvalues~$\lambda_h$ given by
\begin{align}
\label{eq-lambda-h}
\lambda_h \quad \eqdef \quad \sum_{l=0}^t \binom{t}{l} \left( \frac{2\epsilon}{d}\right)^{l}(1-2\epsilon)^{t-l} \, n_{l,h}
\end{align}
for~$h \in \naturals_{d+1} \intersect \naturals_{t+1}$. The eigenvalue~$\lambda_h$ has multiplicity~$\binom{d}{h}$.
\end{lemma}
\begin{proof}
Let~$\sigma^{B}$ be the quantum state after applying Hadamard operator~$H^{\tensor t}$ on register~$L$, i.e.,
\begin{IEEEeqnarray*}{rCl}
    \sigma^{B} \quad  
    & \coloneqq & \quad \frac{1}{2^d} \sum_{a \in \{0,1\}^d} \ \sum_{u,v \in \naturals_{d+1}^{t}}
    \sqrt{D(u) D(v)} \; \ketbra{u}{v}^{I} \tensor H^{\tensor t} 
    \ketbra{a_{u}}{a_{v}}^{L} H^{\tensor t} \\
    & = & \quad \frac{1}{2^t 2^d} 
    \sum_{u,v \in \naturals_{d+1}^{t}} \sqrt{D(u) D(v)} \; \ketbra{u}{v}^{I} \tensor \sum_{a \in \{0,1\}^d} \ \sum_{x,y\in \{0,1\}^t} ( -1)^{x \cdot a_{u} + y\cdot a_{v}} \ketbra{x}{y}^{L} \enspace.
\end{IEEEeqnarray*}
We have
\[
x \cdot a_u \quad = \quad \sum_{i \in [t] \,:\, x_i = 1} a_{u_i}
    \quad = \quad \sum_{j = 1}^d \parity( u_x)_j \cdot a_j 
    \quad = \quad \parity( u_x) \cdot a \enspace,
\]
where~$\parity :\naturals_{d+1}^{t} \rightarrow \{0,1\}^d$ is the parity signature function defined in \Cref{sec-preliminaries}. So, for fixed~$x,y\in \{0,1\}^d$ and~$u , v \in \naturals_{d+1}^{t}$, we have
\begin{equation*}
    \sum_{a \in \{0,1\}^d} 
    ( -1)^{x \cdot a_{u} + y\cdot a_{v}} \quad = \quad
    \begin{cases}
        2^{d} & \text{if} \   \parity (u_{x}) = \parity (v_{y}) \enspace, \\
        0 & \text{otherwise.}
    \end{cases}
\end{equation*}
From this we get
\begin{equation*}
\sigma^{B} \quad = \quad  \frac{1}{2^t} \sum_{x,y\in \{0,1\}^t} 
    \sum_{\substack{u,v \in \naturals_{d+1}^{t} \\ \parity (u_{x}) = \parity (v_{y})}} 
    \sqrt{D(u) D(v)} \; \ketbra{u}{v}^{I} \tensor \ketbra{x}{y}^{L} \enspace.
\end{equation*}
We can identify the eigenvectors of the state~$\sigma^B$ from this expression. We may verify that for each~$b \in \{0,1\}^d$, the vector
\[
\ket{\phi_b} \quad \coloneqq \quad \sum_{x \in \{0,1\}^t} \sum_{\substack{u \in \naturals_{d+1}^{t} \\ \parity (u_{x}) = b } } \sqrt{D(u)} \; \ket{u}\ket{x}
\]
is an eigenvector of~$\sigma^{B}$ with eigenvalue
\begin{IEEEeqnarray*}{rCl}
\lambda_b \quad 
    & \coloneqq & \quad \frac{1}{2^t} \sum_{x \in \{0,1\}^t} \sum_{\substack{u \in \naturals_{d+1}^{t} \\ \parity (u_{x}) = b } } D(u)
\end{IEEEeqnarray*}
whenever~$\size{b} \le t$, where~$\size{b}$ is the Hamming weight of~$b$. The states~$\ket{\phi_b}$ are readily seen to be orthogonal for distinct~$b$. Since
\[
\frac{1}{2^t} \sum_{b \in \set{0,1}^d, \; \size{b} \le t} \density{\phi_b} \quad = \quad  \sigma^B \enspace,
\]
we have all the non-zero eigenvalues of~$\sigma^B$.

The eigenvalues do not seem to be easy to analyse, as they involve the combinatorial quantity~$n_{l,b} \,$, the number of strings of length~$l$ over~$[d]$ with parity signature~$b$. Nonetheless, we show that the aggregate expression we get by considering the entropy of~$\sigma^B$ may be analysed using standard tail inequalities from probability theory. We express the eigenvalues in a form that enables such analysis. 

Let~$\overline{x}$ denote the bit-wise complement of~$x \in \set{0,1}^t$. Then~$D(u) = D(u_x) \cdot D(u_{\overline{x}})$, and for a fixed~$x$
\[
\sum_{\substack{u \in \naturals_{d+1}^{t} \\ \parity (u_{x}) = b } } D(u)
    \quad = \quad \sum_{\substack{v \in \naturals_{d+1}^{\size{x}} \\ \parity (v) = b } } \sum_{ w \in \naturals_{d+1}^{t - \size{x}}} D(v) \cdot D(w) 
    \quad = \quad \sum_{\substack{v \in \naturals_{d+1}^{\size{x}} \\ \parity (v) = b } } D(v) \enspace.
\]
Using this, setting~$r \eqdef \size{x}$, noting that the parity signature of a string~$v \in \naturals_{d+1}^r$ only depends on its non-zero coordinates, and denoting Hamming weight~$\size{v}$ by~$l$, we rewrite~$\lambda_b$ as follows. 
\begin{IEEEeqnarray*}{rCl}
\lambda_b \quad 
    & = & \quad \frac{1}{2^t} \sum_{x \in \{0,1\}^t} 
    \sum_{\substack{v \in \naturals_{d+1}^{|x|} \\ \parity (v) = b } } D(v) \\
    & = & \quad  \frac{1}{2^t} \sum_{r=0}^{t} \binom{t}{r}
    \sum_{\substack{v \in \naturals_{d+1}^{r} \\ \parity (v) = b } } (1-4\epsilon)^{r-|v|} \left( \frac{4\epsilon}{d}\right)^{|v|} \\
    & = & \quad \frac{1}{2^t} \sum_{r=0}^{t} \binom{t}{r} \sum_{l=0}^{r} \binom{r}{l}  (1-4\epsilon)^{r-l} \left( \frac{4\epsilon}{d}\right)^{l} n_{l,|b|} \\
    & = & \quad  \sum_{l = 0}^t \sum_{r = l}^{t - l} \binom{t}{l}  \binom{t-l}{r-l} \frac{1}{2^{t-r}} \left(\frac{1-4\epsilon} 2\right)^{r-l} \left( \frac{4\epsilon}{2d}\right)^{l} n_{l,|b|} \qquad \left( \text{as } \binom t r \binom rl = \binom tl \binom {t-l}{r-l}\right)  \\
	& = & \quad  \sum_{l=0}^t \binom{t}{l} \left( \frac{2\epsilon}{d}\right)^{l} n_{l,|b|}\left(\sum_{r=l}^{t-l} \binom{t-l}{r-l} \frac{1}{2^{t-r}} \left(\frac{1-4\epsilon} 2\right)^{r-l}\right)\\
	& = & \quad  \sum_{l=0}^t \binom{t}{l} \left( \frac{2\epsilon}{d}\right)^{l}(1-2\epsilon)^{t-l} \, n_{l,|b|}  \enspace.
\end{IEEEeqnarray*}
Note that the eigenvalue~$\lambda_{b}$ only depends on the Hamming weight~$|b|$. Thus its multiplicity is~$\binom{d}{|b|}$ and we denote it by~$\lambda_h$ when the string~$b$ has Hamming weight~$h$.
\end{proof}

For convenience, we define~$\lambda_h$ as in \Cref{lem-spectrum1} even when~$h > t$. (Observe that~$\lambda_h = 0$ in this case.) Define~$\mu_h \eqdef \binom{d}{h} \lambda_h \,$, for~$h \in \naturals_{d+1} \,$. This is a distribution over~$\naturals_{d+1} \,$. 

Since the Hadamard operator is unitary, we have~$\vnentropy(B)_{\rho}  =  \vnentropy(B)_{\sigma}$, and therefore
\begin{IEEEeqnarray}{rCl's}
\vnentropy(B)_{\rho} \quad  
& = & \quad \sum_{h=0}^{d} \binom{d}{h} \lambda_h \log 
\frac{1}{\lambda_h} \nonumber\\
& = & \quad \sum_{h=0}^{d} \mu_h \log \frac{1}{\mu_h} + \sum_{h=0}^{d} \binom{d}{h} \lambda_h \log \binom{d}{h}   \nonumber\\
& \leq & \quad \log (d + 1) + \sum_{h=0}^{d} \binom{d}{h} \lambda_h 
\log \binom{d}{h}  
& \qquad (since $\entropy(\mu) \le \log(d+1)$) \nonumber \\
& = & \quad \log (d + 1) +  \sum_{h=0}^{d} \mu_h \log \binom{d}{h}  \enspace. \label{eq-S(B)}
\end{IEEEeqnarray}
Denote the second term in Eq.~\eqref{eq-S(B)} by~$S_{t,d}$. In the rest of the proof, we bound~$S_{t,d}$ by interpreting it as an expectation of~$\log \binom{d}{h}$ with respect to a distribution that is tightly concentrated around its mean. 

\begin{lemma}
\label{lem-entropy-bound}
For~$\nu \le \frac 1 {12} \,$, we have $ S_{t,d} \; \le \;  d \bentropy(6 \nu) + d \exp(- 2\nu d) \,$.
\end{lemma}
\begin{proof}
We define a distribution~$\tD$ on $\naturals_{d+1}$ motivated by the form of $\lambda_h$: $\tD(0) \eqdef 1-2\epsilon$, and~$\tD(i) \eqdef \frac{2\epsilon }{d}$ for $i\in [d]$. Let~$Y_t$ be a random string in~$\naturals_{d+1}^{t}$ each symbol of which is chosen independently according to the distribution~$\tD$. Let~$Z_{tj}$ be the binary random variable corresponding to the parity of the number of occurrences of the symbol~$j\in[d]$ in~$Y_t$. Then~$Z_t \coloneqq \sum_j Z_{tj}$ is the random variable corresponding to the Hamming weight of the parity signature of~$Y_t \,$. 

Observe that~$\mu_h$ is the probability that the parity signature of~$Y_t$ has Hamming weight~$h$, i.e.,~$\mu_h = \Pr[Z_t = h]$. Also,~$|Y_t| \geq Z_t \,$, as the Hamming weight of the parity signature of~$Y_t$ is at most the number of non-zero symbols in~$Y_t \,$. Since~$|Y_t|\sim \mathrm{Binomial}(t, 2\epsilon)$, we have $\expct |Y_t|= 2 \epsilon t = 2\nu d$. By the multiplicative Chernoff bound in Eq.~\eqref{eq:Chernoff-bound}, we have
\begin{equation*}
\label{ham_weight_tail_PAC}	
\Pr[\, |Y_t| \; \geq \; 6\nu d \,] \quad \leq \quad \exp\left(- 2 \nu d \right) \enspace.
\end{equation*}
We use this to bound~$S_{t,d} \,$. Since~$6 \nu d \le d/2$, we have
\begin{IEEEeqnarray*}{rCl}
S_{t,d} \quad & = & \quad  \sum_{h = 0}^d \mu_h \log \binom{d}{h} \\
    & = & \quad  \sum_{h \,<\, 6\nu d} \mu_h \log \binom{d}{h} + \sum_{h \,\ge\, 6\nu d} \mu_h \log \binom{d}{h} \\
    & \le & \quad  d \bentropy(6\nu) + d \Pr[ Z_t \; \ge \; 6\nu d \,] \\
    & \le & \quad  d \bentropy(6\nu) + d \Pr[ \, |Y_t| \; \ge \; 6\nu d \,] \qquad (\text{since } Z_t \le \size{Y_t}) \\
    & \le & \quad  d \bentropy(6\nu) + d \exp( -2 \nu d) \enspace,
\end{IEEEeqnarray*}
as claimed.
\end{proof}

Combining Eq.~\eqref{eq-S(B)} and \Cref{lem-entropy-bound}, we get
\begin{IEEEeqnarray*}{rCl}
	\mi(A:B)_{\rho} \quad = \quad \vnentropy(B)_{\rho} \quad 
	& \leq & \quad \log (d + 1) + d \bentropy(6 \nu) + d \exp(- 2\nu d)
\end{IEEEeqnarray*}
for $\nu \le \frac 1 {12} \,$. This contradicts \Cref{lem-pac-I(A:B)-lb} for small enough~$\nu$ and large enough~$d$.
\end{proof}

\subsection{Sample complexity of agnostic learning}
\label{sec-agnostic-learning}

As in the case of the PAC model, the lower bound for the sample complexity of agnostic quantum learning in Eq.~\eqref{eq:agn-sample-complexity} has two parts. For completeness, we provide the proof of the VC-independent part due to Arunachalam and de Wolf~\cite{AW18-optimal-sample}.
\begin{lemma}
\label{lemma-agn-no-VC}
Let~$\scC$ be a non-trivial concept class. For every~$\delta \in (0,1/2)$,~$\epsilon \in (0,1/4)$, an~$(\epsilon,\delta)$-agnostic quantum learner for~$\scC$ has quantum sample complexity at least~$\Omega(\tfrac{1}{\epsilon^2} \log \tfrac{1}{\delta})$.
\end{lemma}
\begin{proof}
Since~$\scC$ is non-trivial, there exists two distinct concepts~$c_1,c_2 \in \scC$ and an input~$x \in \set{0,1}^n$ such that~$c_1(x) \ne c_2(x)$. Define distributions~$D_{+}$ and~$D_{-}$ as follows. 
\[
D_{\pm}(x,c_1(x)) \quad \eqdef \quad \frac{1\pm 2\epsilon}{2} \qquad \text{and} \qquad D_{\pm}(x,c_2(x)) \quad \eqdef \quad \frac{1\mp2\epsilon}{2} \enspace.
\]
Let~$\ket{\psi_{\pm}} \coloneqq \sqrt{\br{1 \pm 2\epsilon}/2} \; \ket{x,c_1(x)} + \sqrt{\br{1 \mp 2\epsilon}/2} \; \ket{x,c_2(x)}$ be the outputs of agnostic quantum oracles $\QAEX(D_{\pm})$. 

Under the above distributions,~$\opt_{D_\pm}(\scC) = \br{1- 2\epsilon}/2$. Any hypothesis that has error at most~$1/2$ with respect to~$D_+$, i.e., error at most~$\epsilon$ larger than~$\opt_{D_+}(\scC)$ agrees with~$c_1$ at~$x$, and any hypothesis that has error at most~$1/2$ with respect to~$D_-$ agrees with~$c_2$ at~$x$.
So, an~$(\epsilon,\delta)$-agnostic quantum learner for~$\scC$ with sample complexity~$t$ can be used to distinguish~$\ket{\psi_+}^{\tensor t}$ from~$\ket{\psi_-}^{\tensor t}$ with probability at least~$1-\delta$. As in Lemma~\ref{lemma-PAC-no-VC}, we conclude that~$\braket{\psi_+}{\psi_-}^{t} \le 2\sqrt{\delta (1-\delta)}$ whereas~$\braket{\psi_+}{\psi_-} = \sqrt{1- 4\epsilon^2}$. This implies that~$t \in \Omega\br{\tfrac{1}{\epsilon^2} \log \tfrac{1}{\delta}}$.
\end{proof}

Next, we prove the VC-dependent part of the lower bound, using the same ideas as in the proof of Theorem~\ref{thm-PAC-learning} --- via information theory, spectral analysis, and concentration of measure. We refer the reader to that proof for the intuition behind the proof below, as also for some details.
\begin{theorem}
Let~$n \ge 2$, $d \ge 1$, $N \eqdef 2^n$, and let~$\scC \subseteq \set{0,1}^N$ be a concept class with~$\VC(\scC)$ equal to~$d$. There is a universal constant~$d_0$ such that for any~$d \ge d_0$, $\epsilon \in (0,1/4)$, 
and~$\delta \in (0,1/8)$, every~$(\epsilon, \delta)$-agnostic quantum learner for~$\scC$ has sample complexity~$\Omega (d/\epsilon^2)$.
\end{theorem}
\begin{proof}
Suppose there is an~$(\epsilon, \delta)$-agnostic quantum learner~$\scA$ for~$\scC$ using~$t$ quantum examples. Let~$S \eqdef \{s_1,\ldots, s_d\} \subseteq \{0,1\}^n$ be a set of size~$d$ shattered by~$\scC$. For~$a \in \{0,1\}^d$, let~$c_a \in \scC$ be a concept such that~$c_a(s_i) = a_i$ for all~$i \in [d]$. Let~$D_a$ be a distribution over~$\{0,1\}^n \times \{0,1\}$ such that~$D_a(s_i \,,b) \eqdef \left(1+(-1)^{a_i+b} 4\epsilon\right)/2d$ for all~$i\in [d]$ and~$b \in \{0,1\}$, and~$D_{a}(x,b) \eqdef 0$ if~$x \not\in S$. 
Since the distributions~$D_a$ are only supported over the set~$S \times \{0,1\}$, in the rest of the proof, we use~$i$ to mean~$s_i$ for the sake of brevity. Thus we think of~$D_a$ as a distribution over~$[d] \times \set{0,1}$. The corresponding agnostic quantum sample for~$c_a$ is~$\ket{\psi_a} \coloneqq \sum_{i=1}^{d} \sum_{l \in \set{0,1}} \sqrt{D_a(i,l)} \ket{i,l}$. Denote~$t$ copies of~$\ket{\psi_a}$ by~$\rho_a$ and define
\begin{equation}
    \rho^{AB_1\dotsb B_t} \quad \eqdef \quad \frac{1}{2^d} \sum_{a \in \{0,1\}^d} \density{a}^{A} \tensor \rho^{B_1 \dotsb B_t}_{a} \quad 
    = \quad \frac{1}{2^d} \sum_{a \in \{0,1\}^d} \density{a} \tensor \density{\psi_a}^{\tensor t} \enspace.
\end{equation}

Suppose the learner~$\scA$ is given the~$t$ quantum examples~$\rho_a$ in register~$B$. With probability at least~$1-\delta$,~$\scA$ outputs a random) hypothesis~$H_a \in \set{0,1}^N$ satisfying 
\[
\err_{D_a} \! (H_a) \quad \leq \quad \opt_{D_a} \! (\scC) + \epsilon \enspace.
\]
By construction,~$c_a$ is the minimal-error concept from~$\scC$ with respect to the distribution~$D_a$ and~$\err_{D_a} \! (c_a) = (1 - 4 \epsilon)/2 $. We may verify that
\[
\err_{D_a} \! (H_a) \quad = \quad \err_{D_a} \! (c_a) + \frac{4 \epsilon}{d} \sum_{i = 1}^d \Pr[ H_a(i) \neq c_a(i) ] \enspace.
\]
By the agnostic learning criterion the additional error over~$\err_{D_a} \! (c_a) $ is at most~$\epsilon$. Therefore, we have
\[
\frac{1}{d} \sum_{i = 1}^d \Pr[ H_a(i) \neq c_a(i) ] \quad \leq \quad \frac{1}{4} \enspace.
\]
Through the same reasoning as in \Cref{lem-pac-I(A:B)-lb}, we get
\begin{IEEEeqnarray}{rCl}
\label{eq-agn-I(A:B)-lb}
   \mi(A:B)_{\rho} \quad \ge \quad \left(1- \bentropy(3/8) \right) d \enspace.
\end{IEEEeqnarray}

Let~$\alpha \coloneqq \frac{ 1-\sqrt{1-16\epsilon^2}}{2}$.
We show that when~$t = \nu d/2 \alpha$ with~$\nu$ chosen to be a sufficiently small positive constant, the mutual information~$\mi(A:B)_{\rho}$ violates the inequality~\eqref{eq-agn-I(A:B)-lb}. Therefore, we get~$t \in \Omega (d/\epsilon^2)$. 

For~$u \in [d]^t$ and~$l \in \set{0,1}^t$, denote~$\prod_{j = 1}^t D_a(u_j, l_j)$ by~$D_a(u,l)$. We separate the indices and the corresponding bits in register~$B$ into subregisters~$I$ and~$L$, respectively. We then have
\begin{IEEEeqnarray}{rCl}
    \rho^{B} \quad & = & \quad \frac{1}{2^d} 
    \sum_{a \in \{0,1\}^d} \ \sum_{u,v \in [d]^{t}} \sum_{l,k \in \{0,1\}^t}
    \sqrt{D_a(u,l) D_a(v,k)} \, \ketbra{u}{v}^{I} \tensor 
    \ketbra{l}{k}^{L} \enspace.
\end{IEEEeqnarray}

We first derive the spectrum of~$\rho^B$. 
\begin{lemma}
\label{lem-spectrum2}
The state~$\rho^{B}$ has~$\min \set{d+1, t+1}$ non-zero eigenvalues~$\lambda_h$ given by
\begin{align*}
\lambda_h \quad
    & \eqdef \quad \sum_{r=0}^{t}  \binom{t}{r}\frac{n_{r,h}}{d^r} \, \alpha^r (1 - \alpha)^{t-r}
\end{align*}
for~$h \in \naturals_{d+1} \intersect \naturals_{t+1}$. The eigenvalue~$\lambda_h$ has multiplicity~$\binom{d}{h}$.
\end{lemma}
\begin{proof}
Define~$\sigma^{B}$ to be the state obtained by applying the~$t$-qubit Hadamard transformation to register~$L$ in~$\rho^B$. We have
\begin{IEEEeqnarray*}{rCl}
\sigma^{B} \quad 
    & = & \quad \frac{1}{2^d2^t} 
     \sum_{u,v \in [d]^{t}} \ketbra{u}{v}^{I} \tensor \sum_{a \in \{0,1\}^d}  \sum_{\substack{l,k \in \{0,1\}^t \\ x,y \in \{0,1\}^t}} (-1)^{x \cdot l + y \cdot k}
    \sqrt{D_a(u,l) D_a(v,k)} \, \ketbra{x}{y}^{L} \enspace.
\end{IEEEeqnarray*}
We simplify this expression as follows. Consider fixed~$a,u,v,x,y$ as in the expression. Define
\[
\beta_x \quad \coloneqq \quad 
    \frac{1}{(2d)^{t/2}} \left( \sqrt{1 + 4 \epsilon} + \sqrt{1 - 4 \epsilon} \, \right)^{t - \size{x}} \left( \sqrt{1 + 4 \epsilon} - \sqrt{1 - 4 \epsilon} \, \right)^{\size{x}} \enspace,
\]
and note that
\[
\beta_x^2 \quad = \quad 
    \left( \frac{2}{d} \right)^t \alpha^{ \size{x}} (1 - \alpha)^{t - \size{x}} \enspace.
\]
We have
\begin{align*}
\sum_{l \in \set{0,1}^t} (-1)^{x \cdot l} \sqrt{D_a(u,l)} \quad
    & = \quad \prod_{j = 1}^t \sum_{l_j \in \set{0,1}} (-1)^{x_j l_j} \sqrt{D_a(u_j, l_j)} \\
    & = \quad \prod_{j = 1}^t \left( \sqrt{D_a(u_j, 0)} + (-1)^{x_j} \sqrt{D_a(u_j, 1)} \right) \\
    & = \quad \frac{1}{(2d)^{t/2} } \prod_{j = 1}^t (-1)^{a_{u_j} x_j} \left( \sqrt{1 + 4 \epsilon} + (-1)^{x_j} \sqrt{1 - 4 \epsilon} \, \right) \\
    & = \quad (-1)^{a_u \cdot x} \beta_x  \enspace.
\end{align*}
Hence,
\begin{IEEEeqnarray}{rCl}
\nonumber
\sum_{a \in \{0,1\}^d} \sum_{l,k \in \{0,1\}^t} (-1)^{x \cdot l + y \cdot k}
\sqrt{D_a(u,l) D_a(v,k)} \quad 
    & = & \quad \sum_{a \in \{0,1\}^d} (-1)^{a_u \cdot x + a_v \cdot y} \beta_x \beta_y \\
    & = & \quad 
    \begin{cases}
        2^d \beta_x \beta_{y} & \text{if } \parity (u_x) = \parity (v_{y}) \\
        0 & \text{otherwise, }
    \end{cases}
\end{IEEEeqnarray}
and
\begin{IEEEeqnarray*}{rCl}
    \sigma^{B} \quad 
    & = & \quad \frac{1}{2^t} 
     \sum_{ x,y \in \{0,1\}^t} \beta_x \beta_{y}\sum_{\substack{u,v \in [d]^{t} \\ \parity (u_x) = \parity (v_{y}) } } \ketbra{u}{v}^{I} \tensor  \ketbra{x}{y}^{L} \enspace.
\end{IEEEeqnarray*}
We may verify that for each~$b \in \{0,1\}^d$ with~$\size{b} \le t$, the vector
\begin{equation*}
    \ket{\phi_b} \quad \eqdef \quad \sum_{ x \in \{0,1\}^t} \beta_x \sum_{\substack{u \in [d]^{t} \\ \parity (u_x) = b} } \ket{u}\ket{x}
\end{equation*}
is an eigenvector of~$\sigma^{B}$ with corresponding eigenvalue~$\lambda_b$, where
\begin{IEEEeqnarray*}{rCl}
\lambda_b \quad 
    &\coloneqq& \quad \frac{1}{2^t} \sum_{x \in \{0,1\}^{t}} \sum_{\substack{u \in [d]^{t} \\ \parity (u_x) = b}} \beta_{x}^2 \\
    & = & \quad \frac{1}{2^t} \sum_{x \in \{0,1\}^{t}} d^{t-|x|}   \; \beta_{x}^2 \; n_{|x|,|b|} \\
    & = & \quad \sum_{x \in \{0,1\}^{t}} \frac{n_{|x|,|b|}}{d^{|x|}}\; \alpha^{|x|} ( 1 - \alpha)^{t-|x|} \\
    & = & \quad \sum_{r=0}^{t}  \binom{t}{r} \frac{n_{r,|b|} }{d^r} \;  \alpha^r ( 1 - \alpha)^{t - r} \enspace.
\end{IEEEeqnarray*}
The eigenvalue~$\lambda_{b}$ only depends on the Hamming weight~$|b|$. Thus its multiplicity is~$\binom{d}{h}$, and we write~$\lambda_h$ instead of~$\lambda_{b}$ for any string~$b$ with Hamming weight~$h$. Since
\[
\frac{1}{2^t} \sum_{b \in \set{0,1}^d, \; \size{b} \le t} \density{\phi_b} \quad = \quad \sigma^B \enspace,
\]
we have all the non-zero eigenvalues of~$\sigma^B$, and hence of~$\rho^B$.
\end{proof}

For convenience, we define~$\lambda_h$ as in \Cref{lem-spectrum2} even when~$h > t$. (Observe that~$\lambda_h = 0$ in this case.) 
We may now bound the entropy of~$\rho^B$ as
\begin{align}
\nonumber
\vnentropy(B)_{\rho} \quad 
    & \le \quad \log (d + 1) + \sum_{h=0}^{d} \binom{d}{h} \lambda_h \log \binom{d}{h} \\
\label{eq-agn-entropy-bd}
    & = \quad \log (d + 1) + \sum_{r=0}^{t} \binom{t}{r} \alpha^{r} (1 - \alpha)^{t-r} \sum_{h=0}^{d} \binom{d}{h} \frac{n_{r,h}}{d^r} \log \binom{d}{h} \enspace.
\end{align}
Note that~$q(h) \eqdef \binom{d}{h} \tfrac{n_{r,h}}{d^r}$ is a probability distribution over~$h \in [d]$. Furthermore, the Hamming weight of the parity signature of a string of length~$r$ is at most~$r$, i.e.,~$n_{r,h} = 0$ for~$h > r$. Hence, for every~$r \le \tfrac{d}{2}$, we have
\begin{equation*}
    \sum_{h=0}^{d} \binom{d}{h} \frac{n_{r,h}}{d^r} \log \binom{d}{h} \quad \leq \quad \log \binom{d}{r} \enspace.
\end{equation*} 
Since the binomial distribution~$\sB(t,\alpha)$ is concentrated around its mean~$\alpha t$, by the multiplicative Chernoff bound (Eq.~\eqref{eq:Chernoff-bound}) we have
\begin{IEEEeqnarray*}{rCl's}
\sum_{r=\frac{3\alpha t}{2}}^{t} \binom{t}{r}  \alpha^{r} ( 1-\alpha )^{t-r} \quad
    & \le & \quad \exp \br{- \, \frac{\alpha t}{10}} \enspace.
\end{IEEEeqnarray*}
We take~$\nu < 1/2$ so that~$3\alpha t/2 = 3\nu d/4 < d/2$. To bound the second term in Eq.~\eqref{eq-agn-entropy-bd}, we partition the sum over~$r$ into two intervals,~$r < 3\alpha t/2$ and~$r \ge 3\alpha t/2$. Using the bounds derived above, we have
\begin{IEEEeqnarray*}{rCl's}
\vnentropy(B)_{\rho} \quad 
    & \le & \quad \log (d + 1) + \log \binom{d}{\frac{3\alpha t}{2}} + d \; \exp \br{- \, \frac{\alpha t}{10}} \enspace.
\end{IEEEeqnarray*}
Let~$t = \nu d / 2 \alpha$.
If the last term on the right hand side above is at least~$1$, we have~$\alpha t/10 \le \ln d$, i.e.,~$\nu d /20 \le \ln d$. For any fixed positive constant~$\nu$, this does not hold for sufficiently large~$d$. So for sufficiently large~$d$, the last term is at most~$1$ and
\begin{IEEEeqnarray*}{rCl's}
\vnentropy(B)_{\rho} \quad 
    & \le & \quad \log (d + 1) + \log \binom{d}{\frac{3\alpha t}{2}} + 1 \\
    & \le & \quad \log (d + 1) + d \bentropy \! \left( \frac{3\nu}{4} \right) + 1 \enspace.
\end{IEEEeqnarray*}
This violates the bound on mutual information in Eq.~\eqref{eq-agn-I(A:B)-lb} for sufficiently small positive~$\nu$, and large enough~$d$.
\end{proof}

\section{Quantum Coupon Collection}
\label{sec-coupon-collector}

In this section, we study the Quantum Coupon Collector problem, in particular the sample complexity of both its standard and approximation versions.

\subsection{Technical preparations}
\label{sec-prep}

We begin by introducing the notation and technical background on which we build. The spectrum of a quantum state~$\rho^B$ defined formally below plays an important role in the sequel. Most of the section is devoted to presenting properties of the spectrum that may be inferred directly from prior work.

We use the following notation in the remaining subsections. Let~$k,n,t$ be positive integers such that~$1 < k < n$. Let~$\rho_S \eqdef (\density{\psi_S})^{\tensor t}$, where~$\ket{\psi_S}$ is the quantum sample corresponding to the subset~$S \subset [n]$ (cf.\ Eq.~\eqref{eq-qsample}). Let~$\rho^{A B_1 \dotsb B_t}$ be the following classical-quantum state
\begin{equation}
\label{eq-cq-state}
    \rho^{AB_1\dotsb B_t} \quad \eqdef \quad \frac{1}{ \binom{n}{k}} \sum_{ \substack{S \subseteq [n] \\ \size{S} = k}} \density{S}^{A} \tensor \rho_S^{B_1 \dotsb B_t} \quad 
    = \quad \frac{1}{ \binom{n}{k}} \sum_{ \substack{S \subseteq [n] \\ \size{S} = k}} \density{S} \tensor (\density{\psi_S})^{\tensor t} \enspace.
\end{equation}
Let~$B \coloneqq B_1 \dotsb B_t \,$, with each subregister~$B_i$ holding one quantum sample. The learning algorithms we consider are given the~$t$ copies of the quantum sample~$\ket{\psi_S}$ in register~$B$ as input, corresponding to the set~$S$ in register~$A$. The register~$F$ contains the algorithm's output.

The following observation allows us to focus on a smaller range of parameters when analysing the sample complexity of the learning algorithms.
\begin{lemma}
\label{lem-hard-instances}
Suppose that any learning algorithm with probability of error at most~$\delta \in [0,1)$ for the Quantum Coupon Collector problem with parameters~$k,n_0$ (with~$1 < k < n_0$) requires at least~$f(k,n_0,\delta)$ samples. Then the sample complexity of the problem with parameters~$k,n,\delta$ for any~$n \ge n_0$ is also at least~$f(k,n_0,\delta)$. The analogous property also holds for the approximation version of the problem.
\end{lemma}
\begin{proof}
Since any algorithm for the problem (or its approximation version) with parameters~$k,n,\delta$ also solves the problem with parameters~$k,n_0,\delta$, the claims follow.
\end{proof}

We write the state~$\rho^B$ as~$\rho_t^B$ (with subscript~$t$) when we wish to make its dependence on~$t$ explicit. Let~$m \eqdef n - k$. We study the spectrum of~$\rho_t^B$ when~$k > m$; this happens to be the most relevant case. We have
\[
\rho_t^B \quad = \quad \frac {1}{\binom n m} \sum_{S\in [n], |S|=k} \left( \ketbra{\psi_S}{\psi_S} \right)^{\otimes t} \enspace.
\]
Consider the matrix~$C_t$ defined as
\[
C_t \quad \eqdef \quad \frac{1}{{\binom n m}^{1/2}} \sum_{S\in [n],|S|=k} \ket{S} \left( \bra{\psi_S} \right)^{\otimes t} \enspace.
\]
We have 
\begin{align*}
C_t^\adjoint C_t \quad & = \quad \rho_t^B \enspace, \qquad \text{and} \\
C_t C_t^\adjoint \quad & = \quad \frac {1}{\binom n m} \sum_{S,S'\in [n], |S|=|S'|=k} \braket{\psi_{S'}}{\psi_S}^{t} \; \ketbra{S}{S'} \enspace.
\end{align*}
So~$\rho_t^B$ and $C_t C_t^\adjoint$ have the same multiset of non-zero eigenvalues. Moreover, we see that~$C_t C_t^\adjoint$ belongs to the Johnson association scheme with parameters~$n,k$, as~$\braket{\psi_{S'}}{\psi_S}^{t}$ is determined by the value of~$|S\cap S'|$. (For the definition and properties of the Johnson association scheme we use, see Ref.~\cite[Sections~3.2 and~3.3]{ABCKRW20-quantum-coupon-collector}.) Let~$M_t \eqdef C_t C_t^\adjoint$ for~$t \ge 1$.

Let~$M_0 \eqdef {\binom{n}{m}}^{-1} \allone$, where~$\allone$ is the all-$1$ matrix of dimension~$\binom{n}{m} \times \binom{n}{m}$. Let~$D \eqdef \binom{n}{m} C_1 C_1^\adjoint$.  We see that~$M_t = M_{t - 1} \circledcirc D = M_0 \circledcirc D^{\circledcirc t}$ for all~$t \ge 1$, where~`$\circledcirc$' denotes the Hadamard (element-wise) product of matrices with the same dimensions.

The following may be inferred directly from the properties of the Johnson association scheme.
\begin{lemma}
\label{lem-johnson}
There are~$m+1$ orthogonal projection operators~$E_0, E_1, \dotsc, E_m$ defined on the space
\[
\Span \set{ \ket{S} : S \subset [n], \; \size{S} = k}
\]
such that~$E_0 = \frac{1}{\binom n m} \allone$, $\trace(E_s) = \binom{n}{s}  - \binom{n}{s-1}$ for~$s \ge 1$, and~$\sum_{s = 0}^m E_s = \id$. Moreover, for any~$t \ge 0$, the matrix~$M_t$ has spectral decomposition~$M_t = \sum_{s = 0}^m \lambda_{s,t} E_s$ for some non-negative,  not necessarily distinct reals~$\lambda_{0,t}\,, \lambda_{1,t}\,, \dotsc, \lambda_{m,t}\,$.
\end{lemma}
As an immediate consequence, we get the spectrum of~$\rho_t^B$.
\begin{corollary}
\label{lem-spectrum}
For any~$t \ge 1$, the quantum state~$\rho_t^B$ has eigenvalues~$\lambda_{0,t}\,, \lambda_{1,t}\,, \dotsc, \lambda_{m,t}$ (not necessarily distinct, and not necessarily non-zero), and~$0$. The multiplicity of~$\lambda_{s,t}$ is~$l_s$ independent of~$t$, where~$l_s \eqdef \binom{n}{s}  - \binom{n}{s-1}$ for~$s \ge 1$ and~$l_0 \eqdef 1$. Finally, $\sum_{s = 0}^m l_s \lambda_{s,t} = \trace( \rho_t^B) = 1$.
\end{corollary}

For~$j = 0, 1, \dotsc, m$, define
\begin{align*}
p_{j,-1} \quad & \eqdef \quad \frac{j(k-j+1)(m-j+1)}{(n-2j+1)(n-2j+2)k} \\
p_{j,0} \quad & \eqdef \quad \frac k n+\frac{j(n-j+1)(k-m)^2}{nk(n-2j)(n-2j+2)} \\
p_{j,+1} \quad & \eqdef \quad \frac{(k-j)(n-j+1)(m-j)}{(n-2j)(n-2j+1)k} \enspace.
\end{align*}
Further, define~$p_{-1,i} \eqdef 0$ and~$p_{m+1,i} \eqdef 0$ for~$i \in \set{ -1, 0, +1}$, and similarly~$E_{-1} \eqdef 0$ and~$E_{m+1} \eqdef 0$. We use the following properties proven in Ref.~\cite{ABCKRW20-quantum-coupon-collector}.
\begin{lemma}[Lemma~7 and Corollary~8 in Ref.~\cite{ABCKRW20-quantum-coupon-collector}]
For each~$j \in \naturals_{m+1}$, we have 
\[
E_j \circledcirc D \quad = \quad p_{j+1,-1} E_{j+1} + p_{j,0} E_{j} + p_{j-1,+1} E_{j-1} \enspace,
\]
and the numbers~$p_{j,0}, p_{j,-1}, p_{j,+1}$ are non-negative and satisfy~$p_{j,0} + p_{j,-1} + p_{j,+1} = 1$.
\end{lemma}
This property helps us derive a recurrence for the eigenvalues~$(\lambda_{s,t})$. For convenience, define~$\lambda_{-1,t}$ and~$\lambda_{m+1,t}$ to be~$0$ for all~$t \ge 0$.
\begin{lemma}
\label{lem-eigenvalue-recurrence}
We have~$\lambda_{0,0} = 1$, and~$\lambda_{s,0} = 0$ for~$s \in [m]$. For~$t \ge 1$, for all~$s \in \naturals_{m+1}$,
\[
\lambda_{s,t} \quad = \quad p_{s,0} \lambda_{s,t-1} + p_{s,-1} \lambda_{s-1,t-1} + p_{s,+1} \lambda_{s+1,t-1} \enspace.
\]
\end{lemma}
\begin{proof}
We prove this by induction. Recall that~$M_t = \sum_{s = 0}^m \lambda_{s,t} E_s$ for all~$t \ge 0$, and~$M_0 = {\binom{n}{m}}^{-1} \allone = E_0$. So we get the claimed values of~$\lambda_{s,0}$ for all~$s$.

For~$t \ge 0$ we have
\begin{align*}
M_{t+1} \quad = \quad M_t \circledcirc D 
\quad & = \quad \sum_{s=0}^m \lambda_{s,t} E_s \circledcirc D \\
    & = \quad  \sum_{s=0}^m \lambda_{s,t} \left( p_{s+1,-1} E_{s+1} + p_{s,0} E_{s} + p_{s-1,+1} E_{s-1} \right) \\ 
    & = \quad \sum_{s=0}^m \left( p_{s,0} \lambda_{s,t} + p_{s,-1} \lambda_{s-1,t} + p_{s,+1} \lambda_{s+1,t} \right) E_s \\
    & = \quad \sum_{s=0}^m \lambda_{s,t+1} E_s \enspace.
\end{align*}
The recurrence follows. (Since~$p_{0,-1} = p_{m,+1} = 0$, we do not have to consider the boundary points~$0$ and~$m$ separately.)
\end{proof}

\subsection{Properties of the spectrum}
\label{sec-spectrum}

In this section, we develop new properties of the spectrum of the quantum state~$\rho_t^B$ defined in Eq.~\eqref{eq-cq-state} in \Cref{sec-prep}.

In our analysis, the eigenvalues~$\lambda_{s,t}$ always occur in multiples of~$l_s \,$. We derive a recurrence for~$l_s \lambda_{s,t}$ to analyse the quantities we encounter. For convenience, define~$l_{-1} \eqdef 1$ and~$l_{m+1} \eqdef \binom{n}{m+1} - \binom{n}{m}$.
\begin{lemma}
\label{lem-recurrence2}
We have~$l_0 \lambda_{0,0} = 1$, and~$l_s \lambda_{s,0} = 0$ for~$s \in [m]$. For~$t \ge 1$, for all~$s \in \naturals_{m+1}$,
\[
l_s \lambda_{s,t} \quad = \quad p_{s,0} \big( l_s \lambda_{s,t-1} \big) + p_{s-1,+1} \big( l_{s-1} \lambda_{s-1,t-1} \big) + p_{s+1,-1} \big( l_{s+1}\lambda_{s+1,t-1} \big) \enspace.
\]
\end{lemma}
\begin{proof}
The values for~$t = 0$ are straightforward to verify. Multiplying the recurrence relation in \Cref{lem-eigenvalue-recurrence} by~$l_s$ on both sides, we get
\[
l_s \lambda_{s,t} \quad = \quad p_{s,0} \big( l_s \lambda_{s,t-1} \big) + p_{s,-1} \frac {l_s}{l_{s-1}} \big( l_{s-1} \lambda_{s-1,t-1} \big) + p_{s,+1} \frac{l_s}{l_{s+1}} \big( l_{s+1} \lambda_{s+1,t-1} \big) \enspace.
\]	
We may verify by direct calculation that	
\begin{align*} 
p_{s-1,+1} & \quad = \quad p_{s,-1} \frac {l_{s}}{l_{s-1}} \enspace, \qquad \text{and} \\ 
p_{s+1,-1} & \quad = \quad p_{s,+1} \frac {l_{s}}{l_{s+1}} \enspace,
\end{align*}
which gives us the recurrence relation.
\end{proof}

Due to the recurrence for~$l_s \lambda_{s,t}$ and the property that~$p_{s,0} + p_{s,-1} + p_{s+1}=1$, we may interpret~$l_s\lambda_{s,t}$ as a probability arising from a suitably defined random walk. Define random variables~$(W_t : t \ge 0)$ inductively as follows. Let~$W_0 \eqdef 0$, and~$W_{t+1}$ be given via the transition probabilities below. For~$t \ge 0$ and~$s \in \naturals_{m+1}$,
\begin{align*}
\Pr[ W_{t+1} = a \,|\, W_t = s]
    \quad = \quad
    \begin{cases}
	p_{s,0} &\text{if } a = s \\
	p_{s,-1} &\text{if } a = s-1 \\
	p_{s,+1} &\text{if } a = s+1 \\
    0 & \text{otherwise.}
	\end{cases}
\end{align*}
This defines a random walk~$W$ on the points~$\naturals_{m+1}$. (Again, since~$p_{0,-1} = p_{m,+1} = 0$, we do not have to consider the transitions at the boundary points~$0$ and~$m$ separately.) Furthermore, using the recurrence in \Cref{lem-recurrence2} we may verify that~$l_s \lambda_{s,t} = \Pr[W_t = s]$.
 
A similar random walk was also used in Ref.~\cite{ABCKRW20-quantum-coupon-collector}, but we do not know of a direct relation between the two. In particular, we are not sure if there is an intrinsic reason why the recurrences for the two quantities~$l_s\lambda_{s,t}$ and~$\lambda_{s,t}$ involve the same set of coefficients. Due to this connection we may also relate~$\lambda_{s,t}$ to the random walk~$(W_t)$ by changing the initial condition~$W_0 = 0$. We may verify that~$\lambda_{s,t} = \Pr[ W_t=0 \,|\, W_0=s ]$. So, $\lambda_{s,t}$ is the probability that the walk arrives at~$0$ in~$t$ steps when we start at~$s$, whereas~$l_s \lambda_{s,t}$ is the probability for the reverse traversal (the probability that it arrives at~$s$ in~$t$ steps when we start at~$0$).

The relationship between the quantity~$ l_s \lambda_{s,t}$ and the random walk~$W$ helps us derive several useful bounds that we present next. We say that a random walk~$U \eqdef (U_t : t \ge 0)$ on~$\integers$ \emph{dominates\/} another random walk~$U' \eqdef (U'_t : t \ge 0)$ on~$\integers$ if for all~$t \ge 0$ and~$i \in \integers$, we have~$\Pr[ U_t \ge i] \ge \Pr[ U'_t \ge i]$. I.e., for any integer~$i$, the first walk has at least as high a probability of being to right of~$i$ as the second walk. We start with some general conditions under which one random walk on the line dominates another. In their simplest form, the conditions state that a nearest-neighbour walk on the integers that has a higher probability of moving to the right and a lower probability of moving to the left than another similar walk dominates the other walk. 

\begin{restatable}{lemma}{domination}
\label{lem-domination}
Let~$U \eqdef (U_t : t \ge 0)$ and~$U' \eqdef (U'_t : t \ge 0)$ be two possibly time-dependent random walks on the integers that start at~$0$, and move by at most~$1$ in either direction in one time step. For~$i,j \in \integers$, let~$q_{t,i,j} \eqdef \Pr[ U_t = j \,|\, U_{t-1} = i]$, and~$q'_{t,i,j}$ be the analogous transition probability for~$(U'_t)$.
Suppose that for all~$t \ge 0$, for all~$i \in \integers$ reachable by both walks at time~$t$, we have
\[
q_{t,i,i+1} \quad \ge \quad q'_{t,i,i+1} \qquad \text{and} \qquad
q_{t,i,i-1} \quad \le \quad q'_{t,i,i-1} \enspace; 
\]
and for all~$i$ such that~$i$ is reachable by~$U$ and~$i-1$ by~$U'$ at time~$t$, we have
\[
q_{t,i,i+1} + q_{t,i,i} \quad \ge \quad q'_{t,i-1,i} \enspace.
\]
Then~$U$ dominates~$U'$, i.e., for all~$t \ge 0$ and~$i \in \integers$, we have $\Pr[ U_t \ge i] \ge \Pr[ U'_t \ge i]$. Moreover, $\expct[ U_t ] \ge \expct[ U'_t ]$.
\end{restatable}

We defer the proof of this lemma to \Cref{dom_proof}. 

The transition probabilities of the walk~$W$ are quite complicated. However, they may be approximated by simpler quantities, and this leads us to walks resembling the classical coupon collection process. For a positive integer~$n' \ge m$ consider a set of $n'$ coupons out of which~$m$ coupons are ``marked''. As in the classical coupon collector problem, suppose we start with no coupons, and at each step, pick a coupon uniformly at random out of the~$n'$ coupons. Let $\tW(n') \eqdef \Big( \tW_t(n') : t \ge 0 \Big)$ be the random walk on~$\naturals_{m+1}$ where $\tW_t(n')$ is the number of distinct marked coupons coupons collected within~$t$ steps. In more detail, the walk is defined by the following initial condition and transition probabilities, where~$t \ge 0$, and~$s \in \naturals_{m+1}$.
\begin{equation}
\label{eq-approx-walk}  
\begin{split}
\tW_0(n') \quad & \eqdef \quad 0 \enspace, \qquad \text{and} \\ 
\Pr \Big[\tW_{t+1}(n') = a \; \Big|\; \tW_t(n') =s \Big] \quad
    & \eqdef \quad \begin{cases}
        1 - \frac {m-s} {n'} & \text{if } a=s \\
        \frac {m-s}{n'}& \text{if } a=s+1 \\
        0 & \text{otherwise.}
\end{cases}
\end{split}
\end{equation}

We use the coupon-collector-like random walk $W(n')$ for deriving the bounds in the next three lemmas. We would like to obtain the tightest possible results for the Quantum Coupon Collector problem, i.e., those matching the performance of the known algorithms. For this, we require bounds on the spectral quantities in the lemmas for values of the parameter~$m$ ranging from~$1$ up to~$\Theta(n)$. For technical reasons that become apparent in the proofs, we are only able to obtain bounds for~$m$ up to~$\Theta(n / \ln n)$ in some of the lemmas. Nevertheless, this suffices to characterise the leading term of the sample complexity tightly.
\begin{lemma}
\label{lem-expct-ub}
Let~$t \eqdef k\ln m +cn$ for some~$c \in \reals$, and suppose~$m \ln m \le \size{c}n/10$, and~$1 \le m \le n/40$. Then $ \expct[W_t] \leq m - \min \set{\e^{-2c}, \e^{-c/3}}$.
\end{lemma}
\begin{proof}
Let~$W'$ be the random walk given by~$W' \eqdef \tW(n-5m)$. Then, by \Cref{dom_res1}, the walk~$W'$ is an ``optimistic'' approximation of~$W$, i.e., it dominates~$W$, and~$\expct [W'_t] \geq \expct [W_t]$.  Further, by \Cref{lem_ranwal_estimates},
\[ 
\expct [W'_t] \quad \leq \quad m - m\exp \! \left( - \, \frac{t}{n-5m-1} \right) \enspace.
\]
We show in \Cref{lem_eq_frac} that
\begin{align*}
\label{eq-frac}
\frac{t}{n - 5m -1} \quad & \le \quad \ln m + \max \set{2c, c/3} \enspace.
\end{align*}
Using this, we get
\begin{align*}
\expct[W_t] \quad & \le \quad  m - m \exp \! \left( - \, \frac{t}{n-5m-1} \right) \\
    & \le \quad m - m \exp\left(- \ln m - \max \set{2c, c/3} \right) \\
    & = \quad m - \min \set{ \e^{-2c}, \e^{-c/3} } \enspace,
\end{align*}
as claimed.
\end{proof}

\begin{lemma}
\label{lem-expct-lb}
Let~$t \eqdef cn$ for some~$c \ge 0$, and~$2 \le m \le n/10$. Then
\[
\expct[W_t] \quad \geq \quad m \left(1 -\e^{-c} - \frac {cm}{n}\right) \enspace.
\]
\end{lemma}
\begin{proof}
Consider the random walk~$W'' \eqdef \tW(n)$ on~$\naturals_{m+1}$, and an independent random walk~$V''$ on the non-negative integers defined as follows. Let~ $V''_t \eqdef 0$, and for~$t \ge 0$ and $r \ge 0$,
\begin{align*}
\Pr[ \, V''_{t+1} = a \,|\, V''_t = r] \quad & \eqdef \quad
    \begin{cases}
	\frac {m^2}{n^2} & \text{if } a = r+1 \\
    1 - \frac {m^2}{n^2} & \text{if } a = r \\
    0 & \text{otherwise.}
    \end{cases}
\end{align*}
        
The random process~$(W''_t - V''_t)$ may be viewed as a ``pessimistic'' version of~$(W_t)$; indeed, as we show in \Cref{dom_res2}, the walk~$(W_t)$ dominates~$(W''_t - V''_t)$. So, 
\[
\expct[W_t] \quad \geq \quad \expct[W''_t-V''_t]
    \quad = \quad \expct[W''_t] - \expct[V''_t]
    \quad = \quad \expct[W''_t] - \frac {m^2}{n^2} t \enspace.
\]
From \Cref{lem_ranwal_estimates},
\[
\expct [W''_t] \quad \geq \quad m\left(1-\mathrm e^{-\frac t n}\right)
    \quad = \quad m(1-\mathrm e^{-c}) \enspace.
\]
Thus,
\[
\expct[W_t] \quad \geq \quad m \left(1-\e^{-c}\right) - \frac {cm^2}{n} \enspace,
\]
as claimed.
\end{proof}

\begin{lemma}
\label{lem-lmlst-bound}
Let~$t \eqdef k\ln m +cn$ for some~$c \ge 0$, $m \ln m \le cn/10$, and~$1 \le m \le n/40$. Then we have 
\[
\left(1-\frac {tm^2}{n^2}\right) \left(1-\e^{-9c/10}\right) \quad \leq \quad l_m \lambda_{m,t} \quad \leq \quad 1-\frac{\e^{-2c}}{2} \enspace,
\]
where the lower bound additionally requires~$m \ge 2$.
\end{lemma}
\begin{proof}
We use the random walks~$W'$, $W''$, and~$V''$ as defined in the proofs of \Cref{lem-expct-ub,lem-expct-lb}. 

Recall that~$l_m\lambda_{m,t} = \Pr[W_t=m]$; the random variables~$W_t, W'_t$, and~$W''_t - V''_t$ are all bounded by~$m$; by \Cref{dom_res1}, $(W'_t)$ dominates~$(W_t)$; and by \Cref{dom_res2}, when in addition,~$m \ge 2$, the walk~$(W_t)$ dominates~$(W''_t - V''_t)$. So we have
\[
\Pr[W''_t-V''_t = m] \quad \leq \quad l_m\lambda_{m,t} \quad \leq \quad \Pr[W'_t=m] \enspace,
\]
where the lower bound additionally requires~$m \ge 2$.

For the upper bound, we have by \Cref{lem_ranwal_estimates} 
\begin{align*}
\Pr[W'_t=m] \quad & \leq \quad 1 - m \exp \!\left( - \, \frac{t}{n-5m-1} \right) + \frac{m^2}{2} \exp \!\left( - \, \frac{2t}{n-5m} \right) \enspace.\\
\end{align*}
We bound the second term on the right hand side using \Cref{lem_eq_frac}, and the third term by substituting the value of~$t$ to get
\[
\Pr[W'_t=m] \quad \leq \quad 1 - \e^{-2c} + \frac{m^2}{2} \e^{-2\ln m -2c}
    \quad \leq \quad 1 - \frac{\e^{-2c}}{2} \enspace.
\]

For the lower bound in the lemma, we have by \Cref{lem_ranwal_estimates}
\begin{align*}
\Pr[W''_t-V''_t=m] \quad & = \quad \Pr[V''_t=0]\cdot\Pr[W''_t=m] \\
	&\geq \quad \left(1-\frac {tm^2}{n^2}\right) \left(1-m \e^{- t/n}\right)\\
	&\geq \quad \left(1-\frac {tm^2}{n^2}\right) \left(1-\e^{-9c/10}\right) \enspace,
\end{align*}
since~$t = (n - m) \ln m + cn \ge n \ln m - cn/10 + cn$.
\end{proof}

\subsection{Failure of the standard argument}
\label{sec-qcc-it}

In this section, we show that it is not possible to derive optimal lower bound for Quantum Coupon Collection through the standard information-theoretic argument.

Suppose that with probability~$1 - \delta$, for some~$\delta \in (0,1/4]$, a learning algorithm~$\scA$ identifies the set~$S$ in register~$A$,  given the~$t$ copies of the corresponding quantum sample~$\ket{\psi_S}$ in register~$B$. Note that this is equivalent to the assumption that~$\scA$ solves the Quantum Coupon Collector problem with average error~$\delta$ with respect to the uniform distribution over subsets of size~$k$. The algorithm outputs its guess for the set~$S$ in register~$F$. By the Data Processing Inequality, $\mi(A:F) \le \mi(A : B)_\rho$. The standard argument gives a lower bound (say,~$\beta$) for~$\mi(A:F)$ based on the algorithm~$\scA$, and aims to show that if the number of samples~$t$ is~$\order( k \log \min \set{m + 1, k})$, then~$\mi(A : B)_\rho$ is strictly smaller than~$\beta$. This would imply lower bound in Eq.~\eqref{eq-qcc}.

We exhibit adversarial algorithms such that~$\mi(A:F)$ is a constant factor smaller than~$\mi(A:B)$ even when~$t \in \Theta(n)$. So it is not possible to show that~$t$ is asymptotically larger than~$\Omega (n)$ using such an argument. Note that~$\mi(A:B)_\rho = \vnentropy(B)_\rho$. Since the optimal sample complexity is~$\Theta( k \log \min \set{m + 1, k})$, we see that whenever~$\min \set{m, k} \in \upomega(1)$, and~$t \in \Theta(n)$, the uniform ensemble over
\[
\left( \ket{\psi_S}^{\tensor t} : S \subseteq [n], ~ \size{S} = k \right)
\]
has high von Neumann entropy, but there is no measurement which correctly identifies the states with constant probability of success.

We start by describing the aforementioned algorithms.
\begin{lemma}
\label{lem-mi-ub}
Let~$\delta \in (0,1/4]$, and 
let~$\scA'$ be any learner that identifies the set in register~$A$ with probability~$1 - \delta$, given~$t'$ copies of the corresponding quantum sample in the subregister~$B'$ of~$B$. Then there is a learner~$\scA$ which uses~$t$ copies of the sample in register~$B$, with~$t = t' + ck$ for some positive constant~$c$, such that~$\mi(A:F) \leq (1-c_0\delta) \entropy(A)$, where $c_0$ is a positive constant that depends only on $c$.
\end{lemma}
\begin{proof}
Essentially, the learner~$\scA$ uses the~$ck$ additional copies of the quantum sample to check whether the output of~$\scA'$ is correct or not. If~$\scA$ believes the answer is \emph{incorrect\/}, then it ignores the answer, and simply outputs a uniformly random subset of size~$k$.
 
The learner~$\scA$ first runs~$\scA'$ using~$t'$ quantum samples~$\ket{\psi_S}$. Suppose~$\scA'$ outputs~$S'$. Then~$\scA$ performs the projective measurement~$\set{ M, \id - M}$, where~$M \eqdef \density{\psi_{S'}}$, on the remaining quantum samples. If the outcome is~$M$ for all~$ck$ copies, $\scA$ outputs~$S'$. Otherwise, it outputs a uniformly random size-$k$ subset of~$[n]$.

If~$S' = S$ (the set in register~$A$), then with probability~$1$ the outcome is~$M$ for all~$ck$ measurements. Otherwise, suppose that~$d \eqdef \size{S \setminus S'}$. Then, the probability that the outcome is~$M$ in all~$ck$ trials is precisely~$( 1 - d/ k)^{2ck}$, which is at most~$\e^{- 2cd}$. Thus, with probability at least $1-\e^{-2c}$, the learner~$ \scA$ confirms that~$S' \neq S$, and outputs a uniformly random size-$k$ subset. 

There are three cases for the output~$S''$ of~$\scA$:
\begin{enumerate}
\item 
$S' = S$, learner~$\scA$ confirms this via the measurement of the additional~$ck$ samples, and~$S'' = S$.

\item
$S' \neq S$, learner~$\scA$ detects this, and~$S''$ is a uniformly random size-$k$ subset.

\item
$S' \neq S$, learner~$\scA$ fails to detect this, and~$S'' = S'$.
\end{enumerate}
Let the random variable~$E$ denote which of the three cases occurs, and let~$\epsilon_i \eqdef \Pr[E = i]$. Since the probability of correctness of~$\scA'$ is exactly~$1 - \delta$, we have~$\epsilon_2 + \epsilon_3 = \delta$. By the above analysis,~$\epsilon_3 \le \delta \e^{-2c}$.

We have
\begin{align*}
\mi(A:F) \quad
    & = \quad \entropy(F) - \entropy(F | A) \\
    & \le \quad \entropy(F) - \entropy(F | AE) \\
    & \le \quad \log \binom{n}{k} - \epsilon_2 \log \binom{n}{k} \\
    & \le \quad (1 - \delta (1 - \e^{-2c})) \log \binom{n}{k} \enspace,
\end{align*}
where the second inequality follows by noting that
\[
\entropy(F | AE) \quad = \quad \sum_{i \in [3]} \epsilon_i \entropy(F | A, E = i) \quad \ge \quad \epsilon_2 \entropy(F | A, E = 2) \quad = \quad \epsilon_2  \log \binom{n}{k} \enspace. 
\]
So the lemma holds with~$c_0 \eqdef 1 - \e^{-2c}$.
\end{proof}

We show that~$\mi(A:B)$ is larger than this expression even when~$t \in \Theta(n)$, for which we give a lower bound on the mutual information~$\mi(A:B)_\rho \,$.
\begin{lemma}
\label{lem-miab-lb}
Let~$2 \le m \le n/10$, $\kappa \in (0, n/m)$, and~$t \eqdef \kappa n$. Then
\[
\mi(A:B)_\rho \quad = \quad \vnentropy(B)_\rho 
    \quad \geq \quad \left( 1 - \e^{-\kappa} - \frac{\kappa m}{n} \right) \log \binom{n}{m}  - 1 \enspace.
\]
\end{lemma}
\begin{proof}
The spectrum of~$\rho^B$ is given by~$( (l_s, \lambda_{s,t}) : s = 0, 1, \dotsc, m )$, where~$l_s$ and~$\lambda_{s,t}$ are as described in \Cref{sec-prep}. We have
\begin{align*}
\vnentropy(B) \quad & = \quad \sum_{s=0}^m l_s \lambda_{s,t} \log \frac{1}{ \lambda_{s,t}} \\
    & = \quad \sum_{s=0}^m l_s \lambda_{s,t} \log \frac{1}{ l_s \lambda_{s,t}} + \sum_{s=0}^m l_s \lambda_{s,t} \log l_s \\
    & \geq \quad \sum_{s=0}^m l_s \lambda_{s,t} \log l_s \enspace,
\end{align*}
as~$(l_s \lambda_{s,t} : 0 \le s \le m)$ is a distribution. Using the value of~$l_s$,
\begin{align*}
\vnentropy(B) \quad 
    & \geq \quad \sum_{s=0}^m l_s \lambda_{s,t} \log \binom{n}{s} + \sum_{s=0}^m l_s \lambda_{s,t} \log \frac{ \binom{n}{s} - \binom{n}{s-1}}{ \binom{n}{s}} \\
    & = \quad \left[ \sum_{s=0}^m l_s\lambda_{s,t} \frac {\log \binom{ n}{ s}}{\log \binom{n}{m}} \right] \log \binom{n}{m} + \sum_{s=0}^m l_s \lambda_{s,t} \log \frac {n-2s+1}{n-s+1} \enspace.
\end{align*}
Now,
\[
\frac {n-2s+1}{n-s+1} \quad \ge \quad \frac {n-2m+1}{n-m+1} \quad \ge \quad \frac{1}{2} \enspace,
\]
as~$(n - 2s + 1)/(n - s +1)$ is a decreasing function of~$s$. Moreover, by Lemma~\ref{lem-ratio}, 
\[
\frac {\log \binom{ n}{ s}}{\log \binom{n}{m}}
    \quad = \quad \frac {\sum_{i=0}^{s-1} \log \frac{n-i}{i+1} }{ \sum_{i=0}^{m-1} \log \frac{n-i}{i+1}}
    \quad \geq \quad \frac{s}{m} \enspace,
\]
as~$\log \tfrac{n-i}{i+1}$ is decreasing in $i$.
Thus, we have that 
\begin{align*}
\vnentropy(B) \quad 
    & \geq \quad \left[ \frac{1}{m} \sum_{s=0}^m s l_s \lambda_{s,t} \right] \log \binom{n}{m} - 1 \enspace.
\end{align*}
The claim now follows from \Cref{lem-expct-lb}, as~$ \sum_{s=0}^m s l_s \lambda_{s,t} = \expct[ W_t]$.
\end{proof}
Note that by taking~$\kappa$ to be an arbitrarily large constant, we get~$\vnentropy(B)_\rho$ arbitrarily close to~$\entropy(A)$ --- the maximum it can be.

Let~$\delta \eqdef 1/4$.
Taking~$t' \eqdef n/2$ and~$c_0$ to be a suitable constant greater than~$1/2$ in Lemma~\ref{lem-mi-ub}, we see that for the adversarial algorithms described therein, $\mi(A:F) \le \tfrac{7}{8} \log \binom{n}{m}$. On the other hand, taking~$m \eqdef \floor{ \sqrt{n} \,}$, $\kappa$ a sufficiently large constant, and~$n$ sufficiently large, we see that~$\mi(A : B)_\rho \ge \gamma \log \binom{n}{m} = \gamma \log \binom{n}{k}$, for some constant~$\gamma > 7/8$.

\subsection{Sample complexity for approximation}
\label{sec-mismatches}

In this section, we show that the standard argument \emph{does\/} yield a tight lower bound on sample complexity when we are interested in approximating the unknown set. The bound rests on the following estimate.
\begin{lemma}
\label{lem-miab-ub}
Let~$t \eqdef k \ln m + cn$ for some~$c \in \reals$, $m \ln m \le \size{c} n/10$, and~$1 \le m \le n/40$. Then
\begin{align*}
\mi(A:B)_\rho \quad \leq \quad \log \binom{n}{m} - \left(\log n - \log m - \frac{2m}{n} \log \e \right) \min \set{\e^{-2c}, \e^{-c/3}} + \log (m+1) \enspace.  
\end{align*}
\end{lemma}
\begin{proof}
Note that~$\mi(A:B)_\rho = \vnentropy(B)_\rho \,$. The spectrum of~$\rho^B$ is given by~$( (l_s, \lambda_{s,t}) : s = 0, 1, \dotsc, m )$, where~$l_s$ and~$\lambda_{s,t}$ are as given by \Cref{lem-spectrum}. Since~$\sum_{s=0}^m l_s \lambda_{s,t} = 1$, we have
\begin{align*}
\vnentropy(B)_\rho \quad & = \quad \sum_{s=0}^{m} l_s \lambda_{s,t} \log \frac {1}{\lambda_{s,t}} \\
    & \leq \quad \sum_{s=0}^{m} l_s \lambda_{s,t} \log l_s + \log (m+1) \\
    & = \quad \log \binom{n}{m} - \sum_{s=0}^{m} l_s \lambda_{s,t} \log \frac {\binom{n}{m}}{l_s} + \log (m+1) \enspace.
\end{align*}
Since~$l_s \le \binom{n}{s}$, and the ratio~$\tfrac{n-s}{s+1}$ is decreasing in~$s$, and~$1 - z \ge \e^{-2z}$ when~$0 \le z \le \tfrac{\ln 2}{2}$, we have
\begin{align*}
\frac {\binom{n}{m}}{l_s} \quad 
    & \geq \quad \left( \frac {n-s}{s+1} \right) \left( \frac{n-s+1}{s+2} \right) \dotsb \left( \frac {n-m+1}{m}  \right) \\
    & \geq \quad \left( \frac {n-m+1}{m}  \right)^{m - s} \\
    & = \quad \left( \frac {n}{m}  \right)^{m - s} \left( 1 - \frac {m-1}{n }  \right)^{m - s} \\
    & \geq \quad \left( \frac {n}{m}  \right)^{m - s} \exp \! \left( - \, \frac{2 m (m - s)}{n} \right) \enspace.
\end{align*}
Combining these bounds, we get
\begin{align*}
\vnentropy(B)_\rho \quad
    & \leq \quad \log \binom{n}{m} - \left( \log n - \log m - \frac{2 m}{n} \log \e \right) \sum_{s=0}^{m} (m - s) l_s \lambda_{s,t} + \log (m+1) \enspace.
\end{align*}
By Lemma~\ref{lem-expct-ub}, the expectation~$\sum_{s=0}^{m} s l_s \lambda_{s,t} \le m - \min \set{\e^{-2c}, \e^{-c/3}}$, and the claim follows.
\end{proof}

Suppose given~$t$ copies of the sample~$\ket{\psi_S}$ in register~$B$, an algorithm~$\scA$ outputs a size-$k$ subset such that the expected number of mismatches is bounded by~$\Delta$. In other words, if the output is denoted by the random variable~$F(S)$, then~$\expct \size{S - F(S)} \le \Delta$. We first argue that without loss of generality, we may assume that the algorithm is invariant under permutations of~$[n]$ in the sense that~$\Pr[F(S) = S'] = \Pr[F(\pi(S)) = \pi(S')]$ for all pairs of size-$k$ subsets~$S, S' \subseteq [n]$, and permutations~$\pi$ of~$[n]$.
\begin{lemma}
\label{lemma-pi}
Suppose an algorithm~$\scA$ takes some number of quantum samples corresponding to a size-$k$ subset~$S \subseteq [n]$, and outputs a possibly random size-$k$ subset~$F(S)$ of~$[n]$. Then, there is an algorithm~$\scA'$ with output~$F'(S)$ given quantum samples of~$S$ that satisfies the following properties. The algorithm~$\scA'$ is invariant under permutations of~$[n]$, and has the same sample complexity as~$\scA$. Moreover, for every size-$k$ subset~$S$ and uniformly random size-$k$ subset~$A$, the algorithm~$\scA'$ satisfies
\begin{enumerate}
\item $\Pr[F'(S) = S] = \Pr[F(A) = A] \,$, and
\item $\expct \size{S \setminus F'(S)} = \expct \size{A \setminus F(A)} \,$. 
\end{enumerate}
\end{lemma}
\begin{proof}
The algorithm~$\scA'$ samples a uniformly random permutation~$\pi$. Then it maps each copy of $\ket{\psi_S}$ to~$\ket{\psi_{\pi(S)}}$ and runs the algorithm $\scA$ on these samples. Suppose the output of~$\scA$ is $S'$. The algorithm~$\scA'$ outputs $\pi^{-1}(S')$. Permutation invariance follows from the uniform choice of $\pi$, and we may verify that the other properties claimed in the lemma are also satisfied.
\end{proof}

Permutation invariance helps us prove the lower bound for approximation.
\begin{theorem}
\label{thm-qcc-mismatches}
Let~$\Delta \ge 0$ be a constant. Given~$n,k$ and quantum samples for an unknown size-$k$ subset~$S \subseteq [n]$, any algorithm that approximates~$S$ to within~$\Delta$ in expectation uses~$\Omega( k \log \min \set{m + 1, k} )$ samples.
\end{theorem}
\begin{proof}
Let~$\alpha \in (0,1)$ be constant to be specified later. By Lemma~\ref{lem-hard-instances}, it suffices to prove an~$\Omega( k \log (m + 1))$ bound on the sample complexity when~$m \le n^\alpha$ and~$n^{(1 - \alpha)} \ge 2$. (For the details, we refer the reader to \Cref{cor-qcc-lb}, where we make a similar argument.)

Suppose an algorithm uses~$t$ quantum samples and produces a~$\Delta$-approximation to the unknown set. We follow the notation for the registers used by the learning algorithm introduced in \Cref{sec-prep}. By the Data Processing Inequality we have
\[
\mi(A:B)_\rho \quad \geq \quad \mi(A:F) \quad = \quad \entropy(F) - \entropy(F|A) \enspace.
\]
We show that this inequality is violated if~$t = k \ln m + cn$ for a suitable constant~$c \in \reals$. By the Data Processing Inequality,~$\mi(A:B)_\rho$ is non-decreasing in~$t$, so the inequality is also violated for~$t < k \ln m + cn$. We thus conclude a~$k \ln m + cn$ lower bound.

We start with a tight lower bound on~$\mi(A:F)$. By permutation invariance, $\entropy(F) = \entropy(A) = \log \binom{ n}{ m}$ and~$\entropy(F|A) = \sum_{i=0}^m q_i u_i \log \frac{ 1 }{u_i}$, where~$q_i$ is the number of subsets~$S'$ such that for a fixed size-$k$ set~$S$, the pair~$S,S'$ have exactly~$i$ mismatches, and~$u_i$ is the probability of getting a particular such $S'$ as output given samples for~$S$. Note that~$q_i$ and~$u_i$ are independent of~$S$ and~$S'$, and~$q_i = \binom{ k}{ i} \binom{ m}{ i}$.
Now
\[
\log q_i \quad = \quad \log \binom{ k}{ i} + \log \binom{m}{i} \quad \leq \quad i (\log k +\log m) \enspace.
\]	 
Thus
\begin{align*}
\mi(A:F) \quad & = \quad \log \binom{ n}{ m} - \sum_{i=0}^m q_i u_i \log \frac {1}{u_i} \\
    & = \quad \log \binom{ n}{ m} - \sum_{i=0}^m q_i u_i \log q_i - \sum_{i=0}^m q_i u_i \log \frac{1}{ q_i u_i} \\
    & \geq \quad \log \binom{ n}{ m} - \sum_{i=0}^m i q_i u_i (\log k + \log m) - \log (m+1) \\
    & = \quad \log \binom{ n}{ m} - \Delta (\log k + \log m) - \log (m+1) \enspace.
\end{align*}
Combining this with the upper bound on~$\mi(A:B)$ given by  Lemma~\ref{lem-miab-ub} when~$t = k \ln m + cn$, we get
\begin{align*}
\log \binom{n}{m} - \left(\log n - \log m - \frac{2m}{n} \log \e \right) & \min \set{\e^{-2c}, \e^{-c/3}} + \log (m+1) \\
    & \geq \quad \log \binom{n}{m} - \Delta (\log k +\log m) - \log (m+1) \enspace.
\end{align*}
This is equivalent to
\begin{align*}
- \left(\log n - \log m - \frac{2m}{n} \log \e \right) \min \set{\e^{-2c}, \e^{-c/3}} + 2 \log (m+1) + \Delta \log m \quad \geq \quad - \Delta \log k \enspace.
\end{align*}
Using~$1 \le m \le n^\alpha$ and~$k < n$, we get
\begin{align}
\label{eq-mm-bds}
- \left( (1 - \alpha) \log n - \frac{2}{n^{1 - \alpha}} \log \e \right) \min \set{\e^{-2c}, \e^{-c/3}} + 2 \alpha \log n + 2 + \alpha \Delta \log n \quad \geq \quad - \Delta \log n \enspace.
\end{align}
Dividing throughout by~$\log n$ and taking~$n \rightarrow \infty$, we get
\[
\min \set{\e^{-2c}, \e^{-c/3}} \quad \leq \quad \Delta + \frac{2 \alpha}{1 - \alpha} \quad \leq \quad \frac{\Delta + 2 \alpha}{1 - \alpha} \enspace.
\]
Suppose~$\Delta \in [0,1)$. Taking~$\alpha < (1 - \Delta)/3$ and taking~$c \ge 0$ such that
\[
c \quad < \quad \frac{1}{2} \log \frac {1-\alpha}{\Delta + 2 \alpha} \enspace,
\]
and~$n$ large enough, we see that the inequality~\eqref{eq-mm-bds} is violated. (The choice of~$\alpha$ ensures that there is such a~$c \ge 0$.) If~$\Delta \ge 1$, we take~$c < 0$ such that
\[
-c \quad > \quad 3 \log \frac{\Delta + 2 \alpha}{1-\alpha}
\]
for any choice of~$\alpha \in (0,1)$. The right hand side is positive as~$\Delta \ge 1$. Taking~$n$ large enough, we again see that Eq.~\eqref{eq-mm-bds} is violated. This proves the claimed bound. Note that the second order term in the resulting lower bound is~$\Theta(n)$ when~$\Delta < 1$, and~$-\Theta(n)$ when~$\Delta \ge 1$.
\end{proof}

\subsection{Sample complexity of Quantum Coupon Collection}
\label{sec-coupon-collection-lb}

In this section, we prove 
a sharper lower bound on the sample complexity of the Quantum Coupon Collector problem. 

To characterize sample complexity, we bound the decoding error of any learning algorithm by studying the \emph{distinguishability\/} of the ensemble given by a fixed number of samples. The \emph{distinguishability\/} of an ensemble is the probability that an optimal measurement correctly identifies a state selected at random from the ensemble. 
\begin{definition}
Let~$\cF \eqdef \big( ( p_i, \sigma_i ) : \sigma_i \in \qstate( \complex^d), i \in [l] \big) $ be an ensemble of states in which state~$\sigma_i$ occurs with probability~$p_i$. 
We define the \emph{distinguishability\/}~$\dist$ of~$\cF$ as
\[
	\dist(\cF) \quad \eqdef \quad \max_{\textrm{measurement } M} \sum_{i = 1}^l p_i \trace(M_i \sigma_i) \enspace,
\]
where the maximization is over all measurements~$M \eqdef ( M_1, M_2, \dotsc, M_l)$ with~$l$ outcomes.
\end{definition}
We can estimate the distinguishability of an ensemble of states via the generalised Holevo-Curlander bounds stated below~\cite{Holevo79-distinguishability,Curlander79-distinguishability,ON99-converse-channel-coding,Tyson09-distinguishability}. The upper bound on distinguishability was shown by Curlander~\cite{Curlander79-distinguishability} in the case of ensembles of equiprobable \emph{linearly independent\/} pure states. In the ensemble at hand, however, the states are not necessarily linearly independent for arbitrary~$t \ge 1$. The upper bound in the case of ensembles of equiprobable, possibly mixed states was derived by Ogawa and Nagaoka~\cite[Lemma~1]{ON99-converse-channel-coding}. The proof they gave also extends with minor modifications to non-uniform ensembles. The two bounds were proven --- re-proven independently in the case of the upper bound~\cite{Tyson09-erratum} --- by Tyson~\cite[Theorem~10]{Tyson09-distinguishability}. Tyson later gave another proof of the bounds which also generalizes to error-recovery~\cite[Section~III]{Tyson10-error-recovery}. 
\begin{theorem}[generalised Holevo-Curlander bounds; \cite{ON99-converse-channel-coding,Tyson09-distinguishability}]
\label{thm:holevo-curlander}
Let $ \cF \eqdef \big( ( p_i, \sigma_i ) : \sigma_i \in \qstate( \complex^d), i \in [l] \big)$
be an ensemble of~$l$ quantum states. Then the distinguishability of~$\cF$ satisfies
\[
	\left( \hc(\cF) \right)^2 \quad \leq \quad \dist(\cF) \quad \leq \quad \hc(\cF) \enspace,
\]
where
\[
	\hc(\cF) \quad \eqdef \quad \trace \left( \sum_{i = 1}^l p_i^2 \sigma_i^2 \right)^{1/2} \enspace.
\]
\end{theorem}
Using the upper bound on distinguishability above, we first derive a bound for a small range of parameters for which the Quantum Coupon Collector problem appears to be the hardest. (With some additional arguments, we can also prove the theorem below using the original bound due to Curlander.) 
\begin{theorem}
\label{thm-qcc-lb1}
Let~$\delta\in (0,1/40]$ be a constant.
Any algorithm for the Quantum Coupon Collector problem with parameters~$n,k$, and error probability at most~$\delta$ has sample complexity at least
\[
k \ln m + c_0 n
\]
when~$1 \le m \le \delta n$ and~$m\ln m \le c_0 n/20$, where~$c_0 \eqdef \tfrac{1}{2} \ln \big(\tfrac{1 - \delta}{32 \delta } \big)$.
\end{theorem}
\begin{proof}
We follow the notation introduced in Section~\ref{sec-prep}. Suppose that for a constant~$\delta$ as in the statement, with probability~$1 - \delta$, a learner~$\scA$ identifies the unknown~$k$-element subset~$S$ of~$[n]$ in register~$A$, given~$t$ copies of the corresponding quantum sample~$\ket{\psi_S}$ in register~$B$. The algorithm outputs its guess for the set~$S$ in register~$F$. 

Let~$\cE$ be the ensemble consisting of states~$(\density{\psi_{S}})^{\otimes t}$ for a uniformly random size-$k$ subset~$S \subset [n]$. By \Cref{thm:holevo-curlander}, we have
\begin{align*}
1 - \delta \quad & \le \quad \Pr[ F = A] \\
    & \le \quad \hc(\cE) \\
    & = \quad \frac{1}{{\binom n m}^{1/2}} \trace \sqrt{ \rho^B } \enspace. 
\end{align*}
	
Suppose that~$t = k\ln m +cn$, for some parameter~$c$. Since the probability of error may only decrease if the algorithm has access to more quantum samples, we assume that~$c \ge 0$. We show that if~$c$ is ``too small'' a constant, then~$\hc(\cE) < 1 - \delta$, which is a contradiction.

By \Cref{lem-spectrum}, the state~$\rho^B$ has an eigenvalue~$\lambda_{s,t}$ with multiplicity~$l_s$ for each~$s = 0, 1, \dotsc, m$. Using the value of~$l_s$, the Cauchy-Schwarz Inequality, and~$\trace(\rho^B) = 1$, we have
\begin{align*}
\hc(\cE) \quad & = \quad {\binom n m}^{-1/2} \sum_{s=0}^m l_s \sqrt{\lambda_{s,t}} \\
    & = \quad {\binom n m}^{-1/2} \left( \sqrt{l_m} \sqrt{l_m \lambda_{m,t}} + \sum_{s=0}^{m-1} \sqrt{l_s} \sqrt{l_s \lambda_{s,t}} \right) \\
    & \le \quad {\binom n m}^{-1/2} \left( \sqrt{l_m} \sqrt{l_m \lambda_{m,t}} + \left( \sum_{s=0}^{m-1} l_s \right)^{1/2} \left( \sum_{s=0}^{m-1} l_s \lambda_{s,t} \right)^{1/2} \right) \\
    & = \quad {\binom n m}^{-1/2} \left( \left( \binom{n}{m} - \binom{n}{m-1} \right)^{1/2} \sqrt{l_m \lambda_{m,t}} + {\binom{n}{m-1}}^{1/2} ( 1 - l_m \lambda_{m,t} )^{1/2} \right)  \enspace.
\end{align*}
Defining~$ v \eqdef l_m \lambda_{m,t}$ and~$w \eqdef \tfrac{\binom{n}{m-1}}{\binom n m}$, the above inequality simplifies to
\[
\hc(\cF) \quad \leq \quad \sqrt{v(1-w)}+\sqrt {w(1-v)} \enspace.
\]
Intuitively, we have~$ 1-v = 1-l_m\lambda_{m,t} \ll 1$ and~$w = \frac {m}{n-m+1} \ll 1$. So if the upper bound on~$l_m \lambda_{m,t}$ in \Cref{lem-lmlst-bound} were tight, we would have
\[
\hc(\cF) \quad \lesssim \quad  \sqrt{l_m\lambda_{m,t} } \left( 1-\frac{m}{n-m+1} \right)^{1/2}
    \quad \approx \quad 1 - \frac {e^{-2c}}{4} - \frac {m}{2 (n-m+1)} \enspace.
\]
When combined with the lower bound of~$1 - \delta$ on~$\hc(\cE)$, this would give us~$c \in \Omega \big(\ln \tfrac{1}{4 \delta} \big)$. However, the lower bound we have on~$l_m \lambda_{m,t}$ does not suffice for such an argument, and we analyse the bound on distinguishability differently.

Observe that~$w\in (0,1)$ and that as a function in~$x$, the expression~$\sqrt{x(1-w)} + \sqrt {w(1-x)}$ is increasing in the range $[0,1-w]$. Also, from the bound~$m \le \delta n$, we have that~$ w \leq \tfrac{\delta}{1 - \delta}$. Combining these with the upper bound from \Cref{lem-lmlst-bound}, we see that for any~$c \ge c_0/2$,
\begin{align*}
\hc(\cE) \quad & \leq \quad \left(1-\frac {{\e}^{-2c}}2\right)^{1/2} \sqrt{1-w} + \left(\frac {w \e^{-2c}}{2}\right)^{1/2} \\
    & \leq \quad  1 - \frac{\e^{-2c}}{4} + \left( \frac {\delta \e^{-2c}}{2( 1 - \delta)} \right)^{1/2} \\
    & \leq \quad  1 - \frac{\e^{-2c}}{8} \enspace,
\end{align*}
if~$\tfrac{\delta}{1 - \delta} \leq \tfrac{\e^{-2c}}{32}$, or equivalently, $c \le \tfrac{1}{2} \ln\left(\tfrac{1 - \delta} {32 \delta}\right) = c_0$. This is a contradiction as
\[
\hc(\cE) \quad \ge \quad 1 - \delta \quad \ge \quad \left( 1 + \frac{\e^{-2c}}{32} \right)^{-1}
\quad \ge \quad 1 - \frac{\e^{-2c}}{16} \enspace.
\]
This completes the proof.
\end{proof}

Finally, we explain how this implies a more general bound, when~$m$ may be~$\Omega(n/\log n)$. (The~$\order_k(1)$ term in the statement below stands for~$\order(1)$ with respect to the parameter~$k$.)
\begin{corollary}
\label{cor-qcc-lb}
Let~$\delta\in (0,1/40]$ be a constant.
Any algorithm for the Quantum Coupon Collector problem with parameters~$n,k$, and error probability at most~$\delta$ has sample complexity at least
\begin{equation}
\label{eq-cc-lb}
(1 - \order_k(1)) \; k \ln \left( \min\set{k, m + 1} \right)
\end{equation}
whenever~$(\ln k)/k \le \min\set{ c_0 /20, \;\delta \ln 2}$, where~$c_0$ is as defined in \Cref{thm-qcc-lb1}.
\end{corollary}
\begin{proof}
Consider a fixed~$k > 1$. Let~$n_0 \eqdef \floor{ k ( 1 + \beta / \ln k)}$ for a positive constant~$\beta \le 1$ to be specified later, and let~$m_0 \eqdef n_0 - k$. We have~$n_0 > k$ if~$\beta \ge (\ln k)/k$. Secondly, $1 \le m_0 \le \beta k/ \ln k \le \delta k \le \delta n_0$ if~$\beta \le \delta \ln 2$. Finally, 
\begin{align*}
m_0 \ln m_0 \quad & \le \quad \frac{ \beta k}{\ln k} \ln \left( \frac{ \beta k}{\ln k} \right) \\
    & = \quad \beta k + \frac{ \beta k}{\ln k} \ln \left( \frac{ \beta}{\ln k} \right) \\
    & \le \quad \frac{ c_0 n_0}{ 20} \enspace,
\end{align*}
if, in addition,~$\beta \le c_0/20$. A parameter~$\beta$ (depending on the constant~$\delta$) satisfying all three conditions exists as long as~$(\ln k)/k$ is at most~$c_0/20$ and~$\delta \ln 2$.

With~$\beta$ as above, the bound given by Theorem~\ref{thm-qcc-lb1} applies, and the sample complexity of the Quantum Coupon Collector problem with parameters~$n_0, k, \delta$ is at least
\begin{align*}
k \ln m_0 + c_0 n_0 \enspace.
\end{align*}
Up to lower order terms, this is at least~$k \ln k - k \ln\ln k$. By Lemma~\ref{lem-hard-instances}, this lower bound continues to hold for all~$n \ge n_0$. 

For a fixed~$k$, and~$n$ such that~$k < n < n_0$, both~$m/n$ and~$(m \ln m)/n$ decrease as~$n$ decreases. So for such~$n$ we have~$m \le \delta n$ and~$m\ln m \le c_0 n/20$, and the lower bound given by Theorem~\ref{thm-qcc-lb1} holds. Thus we get the bound claimed in Eq.~\eqref{eq-cc-lb} for all~$n > k$.
\end{proof}

Note that \Cref{thm-qcc-lb1} gives us a stronger lower bound when the parameter~$m$ is ``small'' in comparison with~$n$. We lose some of this tightness in extending the lower bound to all parameter regimes.

\bibliographystyle{plain}
\bibliography{bibl}

\begin{thebibliography}{10}

\bibitem{ABCKRW20-quantum-coupon-collector}
Srinivasan Arunachalam, Aleksandrs Belovs, Andrew~M. Childs, Robin Kothari,
  Ansis Rosmanis, and Ronald de~Wolf.
\newblock {Quantum Coupon Collector}.
\newblock In Steven~T. Flammia, editor, {\em 15th Conference on the Theory of
  Quantum Computation, Communication and Cryptography (TQC 2020)}, volume 158
  of {\em Leibniz International Proceedings in Informatics (LIPIcs)}, pages
  10:1--10:17, Dagstuhl, Germany, 2020. Schloss Dagstuhl--Leibniz-Zentrum
  f{\"u}r Informatik.

\bibitem{AdW17-survey}
Srinivasan Arunachalam and Ronald de~Wolf.
\newblock Guest column: A survey of quantum learning theory.
\newblock {\em SIGACT News}, 48(2):41--67, June 2017.

\bibitem{AW18-optimal-sample}
Srinivasan Arunachalam and Ronald de~Wolf.
\newblock Optimal quantum sample complexity of learning algorithms.
\newblock {\em Journal of Machine Learning Research}, 19(1):2879--2878, January
  2018.

\bibitem{AS05-bounds-QL}
Alp At{\i}c{\i} and Rocco~A. Servedio.
\newblock Improved bounds on quantum learning algorithms.
\newblock {\em Quantum Information Processing}, 4(5):355--386, Nov 2005.

\bibitem{B20-phd-thesis}
Shima {Bab Hadiashar}.
\newblock {\em Quantum Compression and Quantum Learning via Information
  Theory}.
\newblock {Ph.D.}\ thesis, University of Waterloo, Waterloo, Ontario, Canada,
  December 2020.

\bibitem{BPP21-generalization}
Leonardo Banchi, Jason Pereira, and Stefano Pirandola.
\newblock Generalization in quantum machine learning: A quantum information
  standpoint.
\newblock {\em PRX Quantum}, 2:040321, November 2021.

\bibitem{B12-PhD-thesis}
Aleksandrs Belovs.
\newblock {\em Applications of the Adversary Method in Quantum Query
  Algorithms}.
\newblock PhD thesis, University of Latvia, Faculty of Computing, 2014.

\bibitem{BEHW89-Learn-VC}
Anselm Blumer, Andrzej Ehrenfeucht, David Haussler, and Manfred~K. Warmuth.
\newblock Learnability and the {Vapnik-Chervonenkis} dimension.
\newblock {\em Journal of the ACM}, 36(4):929--965, October 1989.

\bibitem{BJ99-DNF}
Nader~H. Bshouty and Jeffrey~C. Jackson.
\newblock Learning {DNF} over the uniform distribution using a quantum example
  oracle.
\newblock {\em SIAM Journal on Computing}, 28(3):1136--1153, 1998.

\bibitem{CHCSSCC22-generalization}
Matthias~C. Caro, Hsin-Yuan Huang, M.~Cerezo, Kunal Sharma, Andrew Sornborger,
  Lukasz Cincio, and Patrick~J. Coles.
\newblock Generalization in quantum machine learning from few training data.
\newblock {\em Nature Communications}, 13:4919, 2022.

\bibitem{CT06-info-theory}
Thomas~M. Cover and Joy~A. Thomas.
\newblock {\em Elements of Information Theory}.
\newblock Wiley Series in Telecommunications and Signal Processing.
  Wiley-Interscience, USA, second edition, 2006.

\bibitem{Curlander79-distinguishability}
Paul~Joseph Curlander.
\newblock {\em Quantum Limitations on Communication Systems}.
\newblock PhD thesis, Massachusetts Institute of Technology, Dept. of
  Electrical Engineering and Computer Science, 1979.

\bibitem{FG06-PCT}
J\"org Flum and Martin Grohe.
\newblock {\em Parameterized Complexity Theory}.
\newblock Texts in Theoretical Computer Science: an EATCS series. Springer
  Verlag, Heidelberg, Germany, 2006.

\bibitem{Han16-opt-PAC}
Steve Hanneke.
\newblock The optimal sample complexity of {PAC} learning.
\newblock {\em Journal of Machine Learning Research}, 17(38):1--15, January
  2016.

\bibitem{Haus92-generalize-PAC}
David Haussler.
\newblock Decision theoretic generalizations of the {PAC} model for neural net
  and other learning applications.
\newblock {\em Information and Computation}, 100(1):78--150, 1992.

\bibitem{JS00-distinguishability}
Richard Jozsa and J\"urgen Schlienz.
\newblock Distinguishability of states and {von Neumann} entropy.
\newblock {\em Physical Review A}, 62:012301, June 2000.

\bibitem{KSS94-agnostic}
Michael~J. Kearns, Robert~E. Schapire, and Linda~M. Sellie.
\newblock Toward efficient agnostic learning.
\newblock {\em Machine Learning}, 17(2):115--141, June 1994.

\bibitem{Holevo79-distinguishability}
Alexander~S. Kholevo.
\newblock On asymptotically optimal hypothesis testing in quantum statistics.
\newblock {\em Theory of Probability \& Its Applications}, 23(2):411--415,
  1979.

\bibitem{KNTZ07-pointer-jumping}
Hartmut Klauck, Ashwin Nayak, Amnon {Ta-Shma}, and David Zuckerman.
\newblock Interaction in quantum communication.
\newblock {\em IEEE Transactions on Information Theory}, 53(6):1970--1982,
  2007.
\newblock Preliminary version in Proceedings of the Thirty-Third Annual ACM
  Symposium on the Theory of Computing, 2001.

\bibitem{MU17-probability-computing}
Michael Mitzenmacher and Eli Upfal.
\newblock {\em Probability and Computing: Randomization and Probabilistic
  Techniques in Algorithms and Data Analysis}.
\newblock Cambridge University Press, second edition, 2017.

\bibitem{MR95-randomized-algorithms}
Rajeev Motwani and Prabhakar Raghavan.
\newblock {\em Randomized Algorithms}.
\newblock Cambridge University Press, Cambridge, 1995.

\bibitem{ON99-converse-channel-coding}
Tomohiro Ogawa and Hiroshi Nagaoka.
\newblock Strong converse to the quantum channel coding theorem.
\newblock {\em IEEE Transactions on Information Theory}, 45(7):2486--2489,
  1999.

\bibitem{SG04-sep-Q-C-learning}
Rocco~A. Servedio and Steven~J. Gortler.
\newblock Equivalences and separations between quantum and classical
  learnability.
\newblock {\em SIAM Journal on Computing}, 33(5):1067--1092, 2004.

\bibitem{SB14-ML}
Shai Shalev-Shwartz and Shai Ben-David.
\newblock {\em Understanding Machine Learning: From Theory to Algorithms}.
\newblock Cambridge University Press, Cambridge, UK, 2014.

\bibitem{Tal94-Agn-ub}
Michel Talagrand.
\newblock Sharper bounds for {Gaussian} and empirical processes.
\newblock {\em The Annals of Probability}, 22(1):28--76, 1994.

\bibitem{Tyson09-erratum}
Jon Tyson.
\newblock Erratum: ``{Minimum}-error quantum distinguishability bounds from
  matrix monotone functions: {A} comment on `{Two}-sided estimates of
  minimum-error distinguishability of mixed quantum states via generalized
  {Holevo-Curlander} bounds'\,'' [{J. Math. Phys.} 50, 062102 (2009)].
\newblock {\em Journal of Mathematical Physics}, 50(10):109902, 2009.

\bibitem{Tyson09-distinguishability}
Jon Tyson.
\newblock Two-sided estimates of minimum-error distinguishability of mixed
  quantum states via generalized {Holevo-Curlander} bounds.
\newblock {\em Journal of Mathematical Physics}, 50(3):032106, 2009.

\bibitem{Tyson10-error-recovery}
Jon Tyson.
\newblock Two-sided bounds on minimum-error quantum measurement, on the
  reversibility of quantum dynamics, and on maximum overlap using directional
  iterates.
\newblock {\em Journal of Mathematical Physics}, 51(9):092204, 2010.

\bibitem{val84-PAC}
Leslie~G. Valiant.
\newblock A theory of the learnable.
\newblock {\em Communications of the ACM}, 27(11):1134--1142, 1984.

\bibitem{VC71-VC-dim}
Vladimir~N. Vapnik and Alexey~Y. Chervonenkis.
\newblock On the uniform convergence of relative frequencies of events to their
  probabilities.
\newblock {\em Theory of Probability and its Applications}, 16(2):264--280,
  1971.

\bibitem{VC74-Agn-lb}
Vladimir~N. Vapnik and Alexey~Y. Chervonenkis.
\newblock {\em Theory of pattern recognition}.
\newblock Nauka, Moscow, USSR, 1974.
\newblock In Russian. German Translation: W.~Wapnik \& A.~Tscherwonenkis, {\em
  Theorie der Zeichenerkennung}, Akademie--Verlag, Berlin, 1979.

\bibitem{W18-TQI}
John Watrous.
\newblock {\em The Theory of Quantum Information}.
\newblock Cambridge University Press, May 2018.

\end{thebibliography}

\appendix

\nopagebreak

\section{Coupon-Collector-like random walks}
\label{ran_walk_section}

In \Cref{sec-spectrum}, we analyse properties of the random walk~$(W_t)$ by comparing it with the random walks~$(W'_t)$ and~$(W''_t - V''_t)$ defined in the proofs of \Cref{lem-expct-ub} and \Cref{lem-expct-lb}. Recall that~$W' \eqdef \tW(n-5m)$ and~$W'' \eqdef \tW(n)$, with~$n,m$ and~$\tW(n')$ as in that section. Here we derive a few properties of the random walk~$\tW(n')$ that are used in the analysis. The walk is very similar to the standard coupon collector process.

\begin{lemma}
\label{lem_ranwal_estimates}
The expectation of~$\tW_t(n')$ is given by the exact relation
\[
\expct \big[ \, \tW_t(n') \, \big] \quad = \quad m\left(1-\left(1-\frac 1 {n'}\right)^t\right) \enspace,
\]
and is bounded as
\[ 
m\left(1-\exp \! \left( - \, \frac t {n'-1} \right) \right) \quad \geq \quad \expct \big[ \,\tW_t(n') \, \big] 
    \quad \geq \quad  m\left(1-\exp \! \left( - \, \frac t {n'} \right) \right) \enspace.
\]
The probability that~$ \tW(n')$ reaches~$m$ at time~$t$ is bounded as 
\[ 
1-m \exp \! \left( -\frac t {n'} \right) \quad \leq \quad  \Pr \big[ \, \tW_t(n')=m \, \big] 
    \quad \leq \quad  1-m\exp \! \left( - \, \frac {t}{n'-1} \right) + \frac {m^2} 2 \exp \left( - \, \frac {2t}{n'} \right) \enspace.
\]
\end{lemma}
\begin{proof}
Consider the binary random variables $X_{i,t}$ for $i\in [m]$ and $t\in \naturals$ which is $1$ if the $i$-th coupon is sampled at least once within the first $t$ steps in the random process underlying~$\tW(n')$. Then, 
\[
\tW_t(n') \quad = \quad \sum_{i=1}^m X_{i,t} \enspace.
\]
For a fixed~$t$, the random variables~$X_{i,t}$ are identically distributed for all $i$, and 
\[
\expct[X_{i,t}] \quad = \quad \Pr[X_{i,t}=1] \quad = \quad 1-\Pr[X_{i,t}=0] \quad = \quad 1-\left(1-\frac 1 {n'}\right)^t \enspace.
\]
Thus, $\expct \big[ \, \tW_t(n') \,] = m\left(1-\left(1-\frac 1 {n'}\right)^t\right)$.
Furthermore,
\[
\exp \!\left( - \, \frac {t}{n'-1}\right) \quad \leq  \quad \left(1-\frac{1}{n'}\right)^t
    \quad \leq \quad \exp \!\left( - \, \frac{t}{n'}\right) \enspace,
\]
so we get the inequalities 
\[ 
m\left(1-\exp\left( - \, \frac{t}{n'-1}\right)\right) \quad \geq  \quad \expct \big[ \, \tW_t(n') \, \big] \quad \geq  \quad m\left(1-\exp\left( - \, \frac{t}{n'}\right)\right) \enspace.
\]

For the final part of the lemma, we would like to bound the probability of collecting all~$m$ marked coupons, i.e., of the event that~$\tW_t(n') = m$. By the Union Bound, this probability is bounded from below as
\begin{align*}
\Pr \big[ \, \tW_t(n')=m \, \big] \quad & = \quad \Pr \big[ \, \sum_{i=1}^m X_{i,t} = m \, \big] \\
    & \ge \quad 1-\sum_{i=0}^m \Pr[X_{i,t}=0] \\
    & = \quad 1-m\left(1-\frac 1 {n'}\right)^t \\
    & \geq \quad 1-m\exp\left( - \, \frac{t}{n'}\right) \enspace.
\end{align*}
By considering the first two terms in the Principle of Inclusion-Exclusion, we may also upper bound the probability as
\begin{align*}
\Pr \big[ \, \tW_t(n')=m \, \big] \quad & \le \quad 1-\sum_{i=0}^m \Pr[X_{i,t}=0] + \sum_{1\leq i<j\leq m} \Pr[X_{i,t}=X_{j,t}=0] \\
    & = \quad 1-m\left(1-\frac 1 {n'}\right)^t + \binom m 2\left(1-\frac 2 {n'}\right)^t \\
    & \leq \quad  1-m\exp\left( - \, \frac {t}{n'-1}\right) + \frac {m^2} 2\exp\left( - \, \frac {2t}{n'}\right) \enspace,
\end{align*}
as claimed.
\end{proof}

\section{Random walk domination}

Here, we relate the random walk~$W \eqdef (W_t)$ with the random walks~$W' \eqdef (W'_t)$ and~$Z \eqdef (W''_t - V''_t)$ defined in the proofs of \Cref{lem-expct-ub} and \Cref{lem-expct-lb}. In particular, we show that~$W'$ dominates~$W$, and~$W$ dominates~$Z$. 

We begin with some general conditions under which one random walk dominates another, as first stated in \Cref{sec-spectrum}.
\domination*
\begin{proof}
\label{dom_proof}
The initial condition and the relationship between the transition probabilities allow us to define a coupling between~$(U_t)$ and~$(U'_t)$ from which the stated properties follow. I.e., we define a pair of correlated random walks~$(Q_t, R_t)$ on the integers such that~$(Q_t)$ has the same distribution as~$(U_t)$, and~$(R_t)$ has the same distribution as~$(U'_t)$. Moreover, we ensure that~$Q_t \ge R_t$ for all~$t \ge 0$. As~$ Q_t \ge R_t$, we get
\begin{align*}
\Pr[ U_t \ge i] \quad & = \quad \Pr[ Q_t \ge i] \\
    & \ge \quad \Pr[ Q_t \ge i \text{ and } R_t \ge i ] \\
    & = \quad \Pr[ R_t \ge i] \\
    & = \quad \Pr[ U'_t \ge i] \enspace,
\end{align*}
and
\[
\expct[ U_t ] \quad = \quad \expct[ Q_t ] \quad \ge \quad \expct[ R_t ] \quad = \quad \expct[ U'_t ] \enspace.
\]

It remains to define the coupling. Let~$(H_t : t \ge 0)$ be a sequence of i.i.d.\ random variables distributed uniformly over the unit interval~$[0,1]$. We define~$(Q_t, R_t) \in \integers^2$ inductively using~$(H_t)$ while ensuring that~$\size{Q_t - Q_{t-1}} \le 1$, $\size{R_t - R_{t-1}} \le 1$, and~$Q_{t-1} \ge R_{t-1}$ for all~$t \ge 1$. 

We set~$Q_0 = R_0 \eqdef 0$. For~$t \ge 1$, suppose we have defined~$Q_{t-1}$ and~$R_{t-1}$ satisfying the properties stated above. For~$i,j \in \integers$, conditioned on the event that~$Q_{t-1} = i$ we define
\begin{align*}
Q_t \quad \eqdef \quad \begin{cases}
    i - 1 & \text{if } H_t \le q_{t,i,i-1} \\
    i + 1 & \text{if } H_t > 1 - q_{t,i,i+1} \\
    i & \text{otherwise,}
\end{cases}
\end{align*}
and similarly, conditioned on the event that~$R_{t-1} = j$, we define
\begin{align*}
R_t \quad \eqdef \quad \begin{cases}
    j - 1 & \text{if } H_t \le q'_{t,j,j-1} \\
    j + 1 & \text{if } H_t > 1 - q'_{t,j,j+1} \\
    j & \text{otherwise.}
\end{cases}
\end{align*}

We show~$Q_t \ge R_t$ by induction on~$t$. It holds for~$t = 0$. Assume~$Q_{t-1} \ge R_{t-1}$ for some~$t \ge 1$. Since the two walks move by at most~$1$ at every time step, we need only consider the case when~$Q_{t-1} - R_{t-1} \le 1$. When~$Q_{t-1} = R_{t-1} = i$, the two conditions~$ q_{t,i,i+1} \ge q'_{t,i,i+1}$ and~$ q_{t,i,i-1} \le q'_{t,i,i-1}$ on the transition probabilities ensure that
\begin{itemize}
\item whenever~$R_t$ moves right, $Q_t$ also moves right (i.e., if $R_t = i+1$, then~$Q_t = i+1$), and
\item whenever~$Q_t$ moves left, $R_t$ also moves left (i.e., if~$Q_t = i-1$, then~$R_t = i-1$).
\end{itemize}
When~$Q_{t-1} = i = R_{t-1} + 1$, the condition~$q_{t,i,i+1} + q_{t,i,i} \ge q'_{t,i-1,i}$ on the transition probabilities ensures that whenever~$R_t$ moves right, $Q_t$ either moves right or stays at the same point (i.e., if~$R_t = i$, then~$Q_t \in \set{i, i+1}$). In all the cases, we obtain~$Q_t \ge R_t$.
\end{proof}

The above conditions help us prove concrete domination results. Recall that~$W' \eqdef \tW(n-5m)$, with~$n,m$ and~$\tW(n')$ as defined in \Cref{sec-spectrum}. 
\begin{lemma}
\label{dom_res1}
When~$1\leq m\leq n/40 $, the random walk~$W'$ dominates~$W$.
\end{lemma}
\begin{proof}
We refer the reader to \Cref{sec-spectrum} for the transition probabilities of the walks~$W$ and~$W'$. Since~$W'$ never moves left, we need only compare the probabilities of moving right.

Since~$(k-s)/k \le 1$, $n - s + 1 \le n + 1$, and~$n - 2s + 1 \ge n - 2s \ge n - 2m$, we see that
\[
p_{s,+1} \quad = \quad \frac{(k-s)(n-s+1)(m-s)}{(n-2s)(n-2s+1)k} \quad \le \quad \frac{(n + 1) (m - s)}{(n - 2m)^2} \enspace.
\]
We may verify that the right hand side is at most~$\tfrac{m - s}{ n - 5m}$. By \Cref{lem-domination} the random walk~$W'$ dominates~$W$ and~$ \expct[W'_t] \geq \expct[W_t]$.
\end{proof}

Recall that~$W'' \eqdef \tW(n)$, and that~$V''$ is the random walk on~$\integers$ defined in the proof of \Cref{lem-expct-lb}.
\begin{lemma}
\label{dom_res2}
The random process~$(W''_t-V''_t)$ is a time-dependent random walk, and if~$2\leq m \leq n/10$, the random walk~$(W_t)$ dominates~$(W''_t-V''_t)$.
\end{lemma}
\begin{proof}
For any~$s \in \naturals_{m+1}$ and non-negative integer~$r $ such that~$s + r \le m$ define
\begin{align*}
p''_{s,r,-1} \quad & \eqdef \quad \left(1 - \frac{m-s-r}{n} \right)\frac{m^2}{n^2} \enspace,\\
p''_{s,r,+1} \quad & \eqdef \quad \left(1 - \frac {m^2}{n^2}\right) \frac {m-s-r}{n} \enspace.
\end{align*}
Note that
\begin{align*}
\Pr[ W''_{t+1} - V''_{t+1} = a \,|\, W''_t = s + r, \, V''_t = r] \quad & = \quad
    \begin{cases}
	p''_{s,r,-1} & \text{if } a = s-1 \\
    p''_{s,r,+1} & \text{if } a = s+1 \\
    1 - p''_{s,r,-1} - p''_{s,r,+1} & \text{if } a = s \\
	0 & \text{otherwise.}
    \end{cases}
\end{align*}
Consequently, the transition probabilities of the random walk~$(W''_t - V''_t)$ are functions of the current state and the time step alone:
\begin{align*}
\Pr[ W''_{t+1} - & V''_{t+1} = a \,|\, W''_t - V''_t = s] \\
     & = \quad \sum_{r \ge 0} \Pr[V''_t = r] \cdot \Pr[W''_t = s + r] \cdot \Pr[ W''_{t+1} - V''_{t+1} = a \,|\, W''_t = s + r, \, V''_t = r] \enspace.
\end{align*}
We show below that~$p''_{s,r,-1} \ge p_{s,-1}$, $p''_{s,r,+1} \le p_{s,+1} $, and~$p''_{s-1,r,+1} \le p_{s,0} + p_{s,+1}$. By \Cref{lem-domination}, this suffices to establish the claimed domination property.

We start with~$p''_{s,r,-1} \ge p_{s,-1}$. Since~$p''_{s,r,-1}$ is smallest when~$r = 0$, it suffices to show~$p''_{s,0,-1} \ge p_{s,-1}$, i.e.,
\[
\left(1-\frac {m-s} n\right) \frac {m^2}{n^2} \quad \geq \quad  \frac{s(k-s+1)(m-s+1)}{(n-2s+1)(n-2s+2)k} \enspace.
\]
The inequality holds for~$s = 0$ as~$p_{0,-1} = 0$, so we need only show it for~$s \in [m]$. Multiplying both sides by~$(n-2s+1)(n-2s+2)/m^2$, we see that it is equivalent to
\begin{align}
\label{eq-prob1}
\left( 1-\frac {m-s} n\right) \left(1-\frac {2s-1}n\right)\left(1-\frac {2s-2}n\right) \quad \geq \quad \frac s m \left(1-\frac {s-1}m\right)\left(1-\frac {s-1}k\right) \enspace.
\end{align}
The right hand side is bounded from above by
\[
\frac{s}{m} \left(1-\frac{s-1}m\right) \quad \leq \quad \frac{1}{4} \left(1+\frac{1}{m} \right)^2 \quad \leq \quad \frac{9}{16} \enspace,
\]
using the AM-GM inequality and~$m \ge 2$. The left hand side of Eq.~\eqref{eq-prob1} is bounded from below by
\[
1 - \frac{m-s}{n} - \frac {2s-1}{n} - \frac{2s-2}{n} \quad \geq \quad 1 - \frac{4m}{n} \quad \geq \quad \frac{6}{10} \enspace,
\]
as~$m \le n/10$. This proves Eq.~\eqref{eq-prob1}.

For the same reason as above, we need only show the relation~$p''_{s,r,+1} \le p_{s,+1}$ when~$r = 0$. The resulting inequality is
\[
\left(1-\frac{m^2}{n^2}\right) \frac {m-s}{n} \quad \leq \quad \frac{(k-s) (n-s+1) (m-s)}{(n-2s) (n-2s+1) k} \enspace.   
\]
Multiplying both sides by~$(n-2s)(n-2s+1)/n$, we see that it is equivalent to
\[
\left(1-\frac {m^2}{n^2}\right) \left(1-\frac {2s}n\right) \left(1-\frac {2s-1}n\right) \quad \leq \quad \left(1-\frac {s-1}n\right)\left(1-\frac {s}k\right) \enspace.
\]
This inequality holds as
\[
\left(1-\frac {m^2}{n^2}\right) \leq 1, \quad \left(1-\frac {2s}n\right) \leq \left(1-\frac {s}k\right), \text{ and } \quad \left(1-\frac {2s-1}n\right) \leq \left(1-\frac {s-1}n\right) \enspace.
\]

Finally we show~$p''_{s-1,r,+1} \le p_{s,0} + p_{s,+1} = 1 - p_{s,-1}$. We have~$p''_{s-1,r,+1} \le m/n \le 1/10$. On the other hand, 
\[
p_{s,-1} \quad = \quad \frac{s(k-s+1)(m-s+1)}{(n-2s+1)(n-2s+2)k} \quad \leq \quad \frac{2 m^2 n}{(n-2m)^2 (n-m)} \quad \leq \quad \frac{5}{18}  \enspace,
\]
using~$s \le m$, $k - s + 1 \le n$, $m - s + 1 \le 2m$, and~$m \le n/10$. So~$p''_{s-1,r,+1} \le 1 - p_{s,-1}$ also holds.
\end{proof}

\section{Miscellaneous}

Here we include some technical lemmas used in \Cref{sec-coupon-collector}.
\begin{lemma}
\label{lem-ratio}
For any non-increasing sequence of positive numbers~$a_1,a_2,\dotsc$, and positive integers~$j_1, j_2$ such that~$j_1 \leq j_2$ we have 
\[
\frac {\sum_{i=1}^{j_1}a_i}{\sum_{i=1}^{j_2}a_i} \quad \geq \quad \frac {j_1}{j_2} \enspace.
\]
\end{lemma}
\begin{proof}
It suffices to show that 
\[
j_2\sum_{i=1}^{j_1}a_i \quad \geq \quad j_1\sum_{i=1}^{j_2}a_i \enspace.
\]
View each side as a sum of $j_1 j_2$ elements, given by an appropriate number of repetitions of~$a_i$ for each~$i \in [j_2]$. Note that~$j_1^2$ elements given by~$j_1$ repetitions of each~$a_i$ for $i\in [j_1]$ occur on both sides. Each remaining element on the left hand side is at least as large as every remaining element on the right.
\end{proof}

Let~$k,n$ be positive integers such that~$1 < k < n$, and let~$m \eqdef n - k$.
\begin{lemma}
\label{lem_eq_frac}
If~$t\eqdef k\ln m+cn$ for some $c\in \mathbb R$, $m\ln m\leq |c|n/10$ and $1\leq m\leq n/40$, then
\[ 
\frac {t}{n-5m-1} \quad \leq \quad \ln m +\max \{2c,c/3\} \enspace.
\]
\end{lemma}
\begin{proof}
The claimed inequality is equivalent to
\begin{align*}
\frac{k \ln m + cn}{n - 5m - 1} - \ln m \quad & \le \quad \max \set{2c, c/3} \\
\Longleftrightarrow \qquad  \frac{(4m +1) \ln m + cn}{n - 5m - 1} \quad & \le \quad \max \set{2c, c/3} \enspace.
\end{align*}
Using~$1 \le m \le n/40$, $m \ln m \le \size{c} n/10$, we have
\[
\frac{(4m +1) \ln m}{n - 5m - 1} \quad \le \quad \frac{5m \ln m}{n - 6m} \quad \le \quad \frac{2\size{c}}{3} \enspace.
\]
Further, when~$c \ge 0$, $cn/ (n - 5m - 1) \le cn / (n - 6m) \le 4c/3$, and when~$c < 0$, $cn/ (n - 5m - 1) \le c$. Combining these, we see that the claimed inequality holds. 
\end{proof}

\end{document}